\documentclass[acmtog,nonacm]{acmart} 
\acmJournal{TOG}

\title{Topological Offsets}
\usepackage[english]{babel}
\usepackage{libertinus}
\usepackage{algpseudocodex}
\usepackage{algorithm}
\usepackage{amsfonts}
\usepackage{amsmath}
\usepackage{amsthm}
\usepackage{mathtools}
\usepackage{enumerate}
\usepackage{csquotes}
\usepackage{hyperref}
\usepackage{multirow}
\usepackage{cleveref}
\usepackage{wrapfig}

\usepackage{tikz}
\usetikzlibrary{positioning}
\usetikzlibrary{shapes.geometric, arrows}

\newtheorem{theorem}{Theorem}
\newtheorem{lemma}{Lemma}
\newtheorem{definition}{Definition}

\newcommand{\cS}{\mathcal{S}}
\def\M{M} %

\definecolor{teseoCol}{rgb}{.15, .68, .38}

\begin{document}

\acmSubmissionID{348}

\author{Daniel Zint}
\email{daniel.zint@nyu.edu}
\orcid{0000-0003-4491-1685}
\affiliation{%
  \institution{New York University}
  \country{USA}
}

\author{Zhouyuan Chen}
\email{zc2952@nyu.edu}
\orcid{0009-0007-6493-8734}
\affiliation{%
  \institution{New York University}
  \country{USA}
}

\author{Yifei Zhu}
\email{yz6994@nyu.edu}
\orcid{0009-0008-5081-3188}
\affiliation{%
	\institution{New York University}
	\country{USA}
}

\author{Denis Zorin}
\email{dzorin@cs.nyu.edu}
\orcid{0000-0001-7733-5501}
\affiliation{%
  \institution{New York University}
  \country{USA}
}

\author{Teseo Schneider}
\email{teseoch@uvic.ca}
\orcid{0000-0002-5969-636X}
\affiliation{%
  \institution{University of Victoria}
  \country{Canada}
}

\author{Daniele Panozzo}
\email{panozzo@nyu.edu}
\orcid{0000-0003-1183-2454}
\affiliation{%
  \institution{New York University}
  \country{USA}
}

\begin{abstract}
We introduce \emph{Topological Offsets}, a novel approach to generate manifold and self-intersection-free offset surfaces that are topologically equivalent to an offset infinitesimally close to the surface.
Our approach, by construction, creates a manifold, watertight, and self-intersection-free offset surface strictly enclosing the input, while doing a best effort to move it to a prescribed distance from the input. Differently from existing approaches, we embed the input in a background mesh and insert a topological offset around the input with purely combinatorial operations. The topological offset is then inflated/deflated to match the user-prescribed distance while enforcing that no intersections or non-manifold configurations are introduced. 

We evaluate the effectiveness and robustness of our approach on the Thingi10k dataset,
and show that topological offsets are beneficial in multiple graphics applications, including (1) converting non-manifold surfaces to manifold ones, (2) creating layered offsets, and (3) reliably computing finite offsets.
    
\end{abstract}

\keywords{offsets, meshing}

\ccsdesc[300]{Computing methodologies~Shape modeling}

\begin{teaserfigure}
    \centering
    \includegraphics[width=\linewidth]{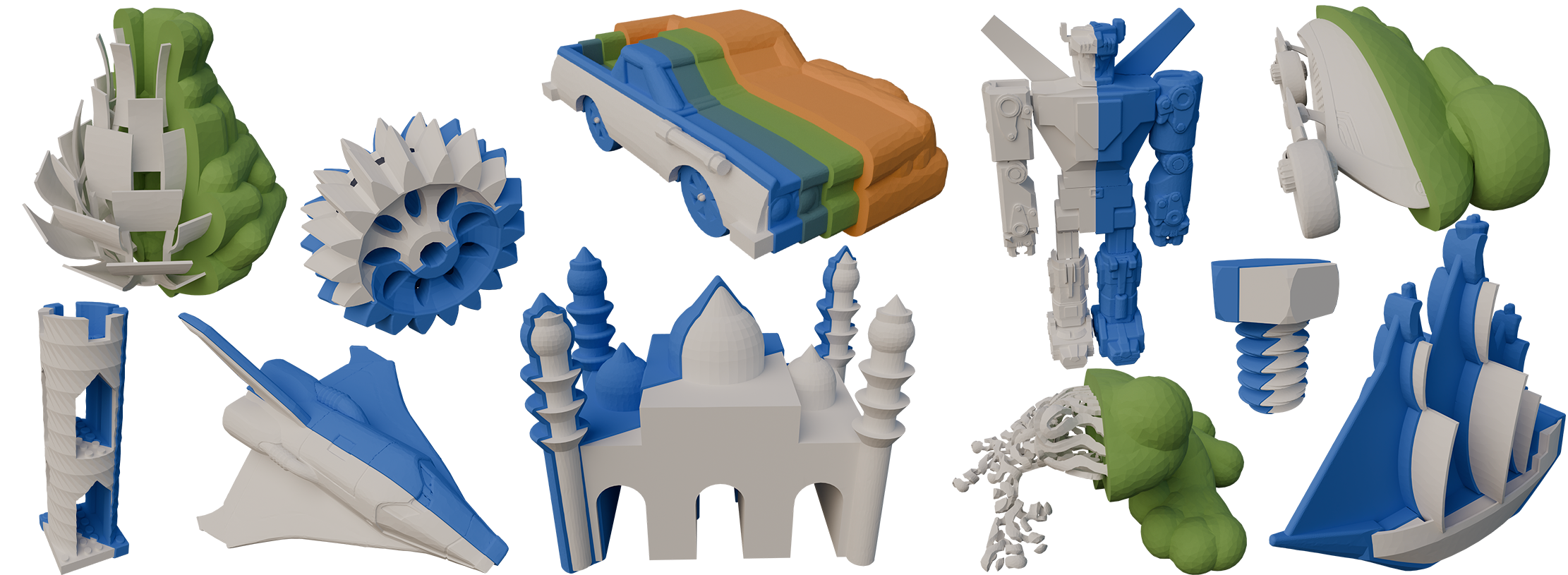}
    \caption{A collection of offsets computed from models in the Thingi10k dataset. Our algorithm provably computes manifold, watertight, and self-intersection-free offsets homeomorphic to an infinitesimally small offset (blue) and, with small changes, can also produce traditional offsets (green), and multiple non-intersecting layers (multicolor).}
    \label{fig:teaser}
\end{teaserfigure}

\maketitle

\section{Introduction}
\label{sec:intro}

Surface offsets, i.e., the regions at a fixed distance from the input surface, are a fundamental modeling tool in graphics and CAD. They are used for designing shapes, computing collision clearances for manufacturing, morphological operators, collision proxies, boundary layers, nested cages, and many more.

Despite their simple definition, their computation is still an unsolved challenge. While exact computation in 2D is possible, it is still an open problem in 3D. This led to a plethora of algorithms computing approximated versions of offsets: unfortunately, they all lose crucial properties such as lack of self-intersections, topological correctness, and geometrical precision and are often restricted to manifold and non-self-intersecting inputs (Section \ref{sec:related}).

Robust and accurate algorithms exist for the special case of offsets with infinite radius, which have the topology of a sphere: shrink-wrapping algorithms deform that infinite offset until it tightly wraps a shape. This is useful for many applications, especially in 3D printing, as it provides a reliable way to inflate shapes with zero thickness or to topologically repair broken meshes. However, it uncontrollably loses the internal holes and handles.

We consider the reciprocal case: we introduce offsets with an infinitesimally small distance from the input and allow them to expand. We call the resulting surface a \emph{topological offset} (\Cref{fig:teaser}, blue). 

We observe and prove that the topology of such an offset is unique and only depends on the input surface. This contrasts traditional offsets, which we will call \emph{finite offsets} from now on (\Cref{fig:teaser}, green), where the topology depends on the offset distance. Interestingly, these offsets are also uniquely defined for non-manifold, non-orientable, and self-intersecting meshes (\Cref{fig:open-non-manifold-self-intersecting-examples}).

We formally define topological offsets (\Cref{sec:continuos_offset}) and introduce a purely topological (and thus unconditionally robust) algorithm to compute it from a tetrahedral background mesh containing the input (\Cref{sec:discrete_topological_offset}). The geometry of the resulting offset is then optimized using a combination of local operations and interior point optimization, with a set of conservative topological and geometrical predicates (accounting for floating point rounding errors) to ensure that the resulting offset keeps the same topology and does not contain self-intersections. 

Our construction is guaranteed to produce offsets with the following properties: (1) no intersections, (2) manifold, (3) same topology as an infinitesimal offset, and (4) strictly enclosing/containing the input (\Cref{fig:fertility_topo_vs_finite}).

\begin{figure}
    \centering\footnotesize
    \includegraphics[width=0.3\linewidth]{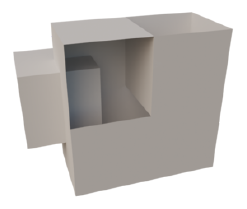}\hfill
    \includegraphics[width=0.3\linewidth]{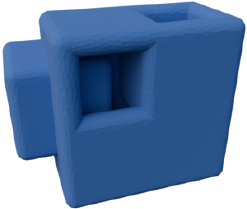}\hfill
    \includegraphics[width=0.3\linewidth]{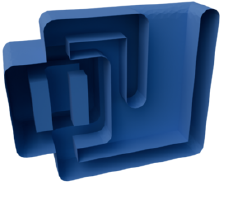}
    \parbox{0.32\linewidth}{\centering input}\hfill
    \parbox{0.64\linewidth}{\centering topological offset}
    \caption{Our method generates manifold meshes with topological guarantees even if the input (left) is open, non-manifold, non-orientable, and self-intersecting. The right image is a cut-view of the offset.}
    \label{fig:open-non-manifold-self-intersecting-examples}
\end{figure}

\begin{figure}
    \centering\footnotesize
    \includegraphics[width=0.32\linewidth]{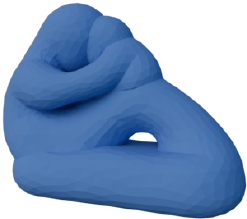}\hfill
    \includegraphics[width=0.32\linewidth]{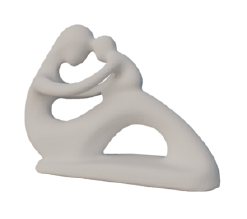}\hfill
    \includegraphics[width=0.32\linewidth]{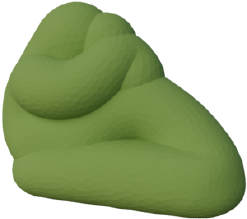}
    \parbox{0.32\linewidth}{\centering topological offset}\hfill
    \parbox{0.32\linewidth}{\centering input}\hfill
    \parbox{0.32\linewidth}{\centering finite offset}
    \caption{Topological offsets (blue) always have the same topology for any given input (white), while the topology of finite offsets (green) depend on the offset distance}
    \label{fig:fertility_topo_vs_finite}
\end{figure}

\begin{wrapfigure}{r}{0.15\linewidth}
    \includegraphics[width=1\linewidth,trim={0.8cm 0.5cm 0cm 0.5cm}]{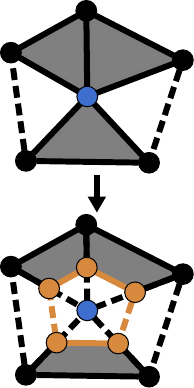}
\end{wrapfigure}
Topological offsets enjoy wide practical applicability (Section \ref{sec:results}): 
Non-manifold simplices can be removed by generating a topological offset around them and tagging all simplices in the offset region as outside (see inset). Self-intersections can be repaired by replacing a self-intersecting surface with its topological offset, which is always intersection-free. 
With a minor modification, topological offsets become finite ones while enjoying the same robustness, ensuring manifoldness and lack of self-intersections for any offset distance. 

We compare our modified version that produces finite offsets against the state-of-the-art (finite) offset method \cite{zint2023feature}: while being robust, this approach produces intersecting offsets on about 5\% of the models, hindering their usability in downstream applications. We also analyze our topological offsets on a large data set but we are not aware of any method that produces similar results that we could compare against.

Our major contributions are:
\begin{enumerate}
    \item We formally define \emph{topological offsets}.
    \item We introduce a robust algorithm to compute them and control their geometry.
    \item We show that they can be converted into \emph{finite offsets}, while ensuring manifoldness and lack of intersection.
    \item We produce layered offsets that strictly enclose each other by repeatedly adding topological or finite offsets.
    \item We introduce (topological and finite) offsets with spatially varying distances.
    \item We apply the topological offset to remove non-manifold edges and vertices while keeping the topological and geometrical changes minimal.
    \item 
    We provide an open-source reference implementation\footnote{https://github.com/wildmeshing/topological-offsets} to ensure the reproduction of our results and foster the adoption of this new offset type.
\end{enumerate}

\section{Related Work}
\label{sec:related}

\paragraph{Voronoi Diagram.}

Generating an exact offset can be considered a sub-problem of generating a generalized Voronoi diagram. While this approach is feasible in 2D, with robust algorithms for computing 2D Voronoi diagrams \cite{cgal:k-sdg2-22b}, the reliable generation of such a diagram in 3D is still an open problem~\cite{yap2012towards,boada2008approximations}. Hemmer et al. compute the exact Voronoi diagram for arbitrary lines in 3D~\cite{hemmer2010constructing}. Aubry et al. use the generalized spherical Voronoi diagram around vertices (which is 2D and therefore easier to compute) to extract boundary layers ~\cite{aubry2017boundary}. However, this method cannot handle self-intersections of the boundary layer, limiting its use to shapes where two non-adjacent elements are further away than twice the offset distance.

\paragraph{Discrete Offset and Morphological Operations}

A volume can be discretized as a collection of voxels, and a discrete offset can be defined using morphological operations \cite{suriyababu2023towards}, which are robust and efficient. However, using a grid (uniform or adaptive) introduces staircase artifacts and inherently limits feature size, as the memory and computation cost to store the data is high, even when Dexel data structures are used \cite{chen2019half}. Other methods sidestep that issue by performing morphological operations on point clouds \cite{calderon2014point}. However, the conversion into a point cloud induces a sampling error. \cite{sellan2020opening} proposes an approach for performing opening and closing operations on meshes without performing dilation and erosion explicitly.
It does not generate offsets. 

\paragraph{Approximate Distance Offsets}

A popular compromise between \emph{exact} and \emph{discrete} offsets is the use of uniform grids, adaptive grids, or particles to discretize a distance field from the input surface \cite{qu2004feature,liu2010fast,pavic2008high,zint2023feature,wang2013gpu,meng2018efficiently} and extract the offset as an isosurface \cite{lorensen1987marching,ju2002dual}. After the surface is extracted, a remeshing procedure is applied to improve mesh quality, reduce element count, and remove self-intersections \cite{botsch2004remeshing}.
The advantage of these approaches is their efficiency in computation cost and memory, as they rely on decades of work on isosurface extraction to get the initial offset with high resolution, and they perform the expensive mesh optimization only on the resulting, possibly adaptive, surface. However, similarly to discrete methods, they cannot guarantee that the extracted offset is homeomorphic to the exact offset. 
Additionally, post-processing the extracted surface might introduce self-intersections \cite{zint2023feature}.
Another approach is to construct linear approximations for each input triangle per cell and compute plane intersections \cite{Wang2024pco}. This approximation is sufficient to produce sharp features in convex regions but the offset might intersect the input for small offset distances.

Our approach computes an offset guaranteed to be manifold, intersection-free, and with the unique topology of an infinitesimal offset. Furthermore, our method can also compute finite offsets and still enjoys the guarantee of producing a manifold and intersection-free output.

\paragraph{Approximate Minkowski Sums}

An alternative approximation is to define the offset as the Minkowski sum of an input surface with a discretized sphere \cite{varadhan2004accurate,campen2010exact,campen2010polygonal,martinez2015chained}. These methods create very dense meshes when using an accurate discretization of the sphere, making them a good fit only for applications where a coarse approximation of the offset is sufficient (see Figure 5 of \cite{Wang2020} for an analysis of memory and time used by this type of algorithms). 

Contrarily, the offset distance for topological offsets is adaptive while keeping a small distance error for finite offsets (\Cref{sec:results}).

\paragraph{Shrink-Wrapping}

Shrink-wrapping algorithms \cite{Kobbelt1999,Lee2009,huang2020manifoldplus,Martineau2016ANIF,Jureti2011,suriyababu2023towards} shrink an infinite offset containing the full mesh until it tightly fits the mesh. This approach can be used to repair meshes reliably \cite{Stuart2013,Portaneri2022,Dai2024feature} and it is a popular method to prepare models for 3D printing. The topology of the resulting surface, by construction, ignores internal holes: while this is a desirable property for mesh repair, the resulting surfaces are not offsets.

Our approach builds upon ideas in \cite{Portaneri2022} to use a tetrahedral background mesh, but uses it to build an infinitesimal offset instead.

\paragraph{Enclosing Volumes and Boundary Layers}

The closest works to ours are algorithms to construct enclosing volumes around an input, which are often used for animation cages, shell maps, or interface tracking in graphics \cite{Sacht2015,Jiang2020, porumbescu2005shell,Brodersen:2008,misztal2012topology}, or boundary layers in engineering simulation \cite{loseille2013robust,Garimella2000}. These methods rely on a displacement in the normal direction (either directly, or via a geometric flow) which is not well-defined, in general (\Cref{fig:boxes}). 

\begin{figure}
    \centering\footnotesize
    \includegraphics[width=.49\linewidth]{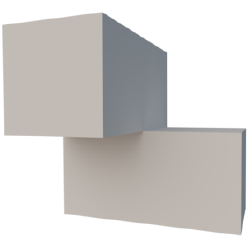}\hfill
    \includegraphics[width=.49\linewidth]{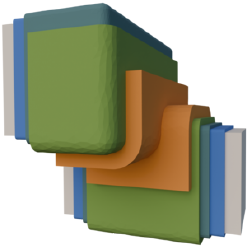}
    \parbox{0.32\linewidth}{\centering input}\hfill
    \parbox{0.32\linewidth}{\centering offset layers}
    \caption{The corner between the two boxes does not have a unique vertex normal, i.e., there is no direction in which the vertex can be offset without intersecting the input. Our method does not rely on vertex normals and can, therefore, handle this case properly. See \cite{Jiang2020} for a more detailed discussion.}
    \label{fig:boxes}
\end{figure}

Instead of using normals, other approaches construct cages from volumetric representations of the input. \cite{calderon2017bounding} rely on voxel grids, that are known to be fast and efficient, but they are not guaranteed to capture the input topology correctly, as input surfaces can be arbitrarily close.
A robust method for constructing cages is presented in \cite{guo2024robust}. The method requires closed and manifold meshes without self-intersections as input, and constructs cages that are guaranteed to be homeomorphic to the input. Like our method, it computes the cage topology using a tetrahedral embedding of the input. However, the tetrahedral mesh is uniformly refined and eventually discarded and therefore no longer available for downstream applications. In contrast, our method keeps the tetrahedral mesh and ensures that it stays inversion-free. Additionally, we only perform local subdivisions instead of a uniform refinement, resulting in a much coarser initial mesh (\Cref{fig:robust-cages}). Finally, our method can also handle open, non-manifold, and non-orientable input with self-intersections (\Cref{fig:open-non-manifold-self-intersecting-examples}).

\begin{figure}
    \centering\footnotesize
    \includegraphics[width=\linewidth]{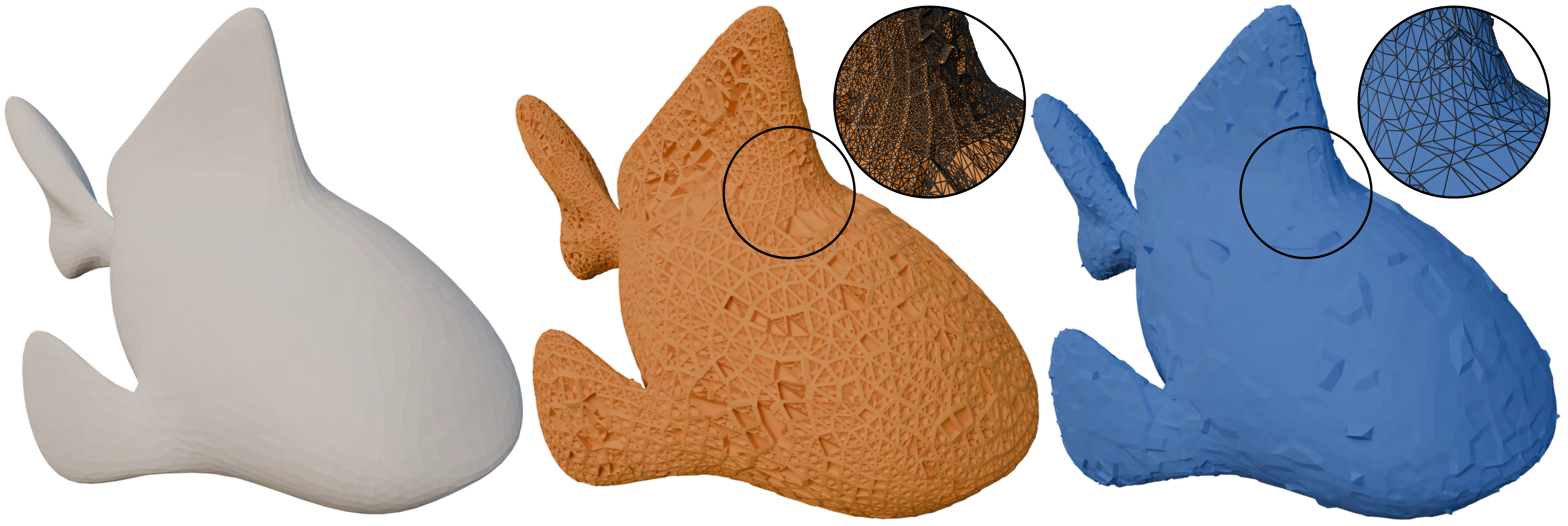}
    \parbox{0.3\linewidth}{\centering input, $\#t = 7\,238$}\hfill
    \parbox{0.3\linewidth}{\centering \cite{guo2024robust}, $\#t = 417\,790$}\hfill
    \parbox{0.3\linewidth}{\centering ours, $\#t = 23\,726$}
    \caption{Our method utilizes local subdivisions instead of a uniform refinement, leading to merely a fraction of triangles. Both, our method and \cite{guo2024robust}, are guaranteed to construct an offset that is homeomorphic to a closed and manifold input mesh without self-intersections.}
    \label{fig:robust-cages}
\end{figure}

\paragraph{Multi-Material Remeshing}

Our method processes a tetrahedral background mesh whose faces represent the input geometry: this is useful to avoid self-intersections in the offset without requiring explicit collision checks. Our algorithm uses a multi-material mesh optimization similar to \cite{FTB:2016:MVR,cgal:tftb-tr-23b}, relying on the multi-material link condition proposed in \cite{Thomas2011}. We provide more details on our remeshing algorithm in \Cref{sec:optimization}.

\section{Overview}

We introduce an algorithm for creating topological offsets. Unlike many other offset methods, we rely on a volumetric meshing algorithm to embed the input into a tetrahedral background mesh before processing: this increases our running time but provides strong guarantees on topology, lack of self-intersections, and termination. Our guarantees hold even when our algorithm is implemented using floating point arithmetic (Section \ref{sec:robustness}). 

\paragraph{Input} The input to our algorithm is a simplicial complex (e.g., a non-manifold triangular surface) embedded within a manifold tetrahedral background mesh embedded in $\mathbb{R}^3$ without inverted elements, plus an offset distance $\delta$. The manifoldness condition is only required on the background mesh, not on the input. Additionally, we require that the simplicial complex 
is in the interior of the tetrahedral background mesh, i.e., there is a layer of tetrahedra completely enclosing the input.

\paragraph{Output} The output of our algorithm is a tetrahedral mesh with the offset (and input) embedded. The embedded triangle mesh of the offset is guaranteed to be manifold, intersection-free, and watertight. Our algorithm strives to improve this mesh's quality and keep a distance $\delta$ from the input (Section \ref{sec:results}). We note that the tetrahedral mesh is useful in many downstream applications and is easy to discard if unnecessary.

\paragraph{Summary}
We first introduce the theoretical concept of topological offsets and prove that their topology is unique for a sufficiently small $\epsilon$ (Section \ref{sec:continuos_offset}). This proof requires us to define a locality condition that strongly relies on a specific type of mesh embedding that we call simplicial embedding~(\Cref{def:simpembedding}).
We then propose an algorithm to compute them and optionally improve their quality (\Cref{sec:discrete_topological_offset}). Crucially, our construction \emph{does not require selecting a sufficiently small $\epsilon$}: it is purely topological and provably produces an offset homeomorphic to an offset with infinitesimally small $\epsilon$ (\Cref{thm:homeomorphic}).

\section{Infinitesimal and Topological Offset}
\label{sec:continuos_offset}

\begin{figure}
    \centering
    \includegraphics[width=1\columnwidth]{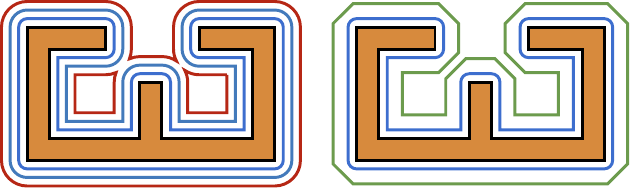}
    \caption{An infinitesimal offset family (blue curves). If $\epsilon$ becomes too large, the offset changes topology (red curve). We are interested in a topological offset (green curve) that admits a continuous bijective map to the infinitesimal one.}
    \label{fig:smooth}
\end{figure}

Consider a smooth, simple, manifold, closed surface $C$ with a single connected component, and bounded curvature embedded in a domain $\Omega$. A finite offset $\mathcal{O}(C,\epsilon)$ is the set of points at a distance $\epsilon > 0$ from the surface,
$$\mathcal{O}(C,\epsilon) = \{x \in \Omega  \mid \|c(x) - x\|_2 = \epsilon\},$$
where $c(x)$ is the function returning the point of $C$ closest to the point $x$. We can always find an $\epsilon_C$ such that the offset does not intersect the medial axis of $C$, and thus, for such an $\epsilon_C$, the function $c(x)$ is bijective. By repeating the same construction for a simplicial complex $S$ (\Cref{fig:smooth}), an $\epsilon$ for which the function $c$ is bijective does not exist anymore since the medial axis extends to the sharp points of $S$. However, as we will show in the following, there is still an $\hat{\epsilon}$ such that all offsets with a distance smaller than $\hat{\epsilon}$ have the same topology, are manifold, and are free of self-intersections.

\begin{definition}[$\hat{\epsilon}$ Infinitesimal Offset Family]
 We call $\bar{\mathcal{O}}(S,\hat{\epsilon})$ the infinitesimal offset family of a simplicial complex $S$ 
 the 1-parameter family of offsets parametrized by a distance parameter $\epsilon \in (0,\hat{\epsilon})$. 
\end{definition}

To streamline the explanation, we denote with $\mathcal{S}_\epsilon(s)$ the $\epsilon$-inflation of the simplex $s$, i.e., the Minkowski sum of $s$ and the ball of radius $\epsilon > 0$. Similarly, we use $\mathcal{S_\epsilon}(S)$ to denote the $\epsilon$-inflation of all simplices in~$S$.

\begin{theorem}[Infinitesimal Offsets]
For a simplicial complex $S$, there is an $\hat{\epsilon} > 0$ such that all offset surfaces $\partial \mathcal{S_\epsilon}(S), 0 < \epsilon \leq \hat{\epsilon}$ are manifold surfaces, assuming that faces of $S$ are in a general position.
\label{thm:infinitesimal}
\end{theorem}
\begin{proof}

The distance function to a convex set is $C_1$ \cite{rockafellar2009variational}.  The distance function to a union of a finite number of convex sets is $\min( d_0 \ldots d_n)$  where $d_i$ is the distance to $i$-th convex set, so clearly piecewise $C_1$ and Lipschitz, as $\min$ is Lipschitz and the composition preserves the Lipschitz property. The offsets are the level sets of this function. 
 
The specialization of Clarke's theorem \cite{Borwein2006} to the scalar functions  $f:  \mathbb{R}^3 \rightarrow \mathbb{R}$
provides a criterion when the resulting offset is a manifold.   Let $U \in \mathbb{R}^3$ be an open neighborhood of a point $p = (x,y,z)$, and let $c = f(p)$.   If there is a matrix $2 \times 3$ matrix $B$, such that for any $A \in \partial f(p)$  (the generalized differential of $f$) The matrix  $\left[\begin{array}{c}A\\ B
\end{array}\right]$ is an invertible $3 \times 3$ matrix,  then the level set  $c = f(p)$ is locally a Lipschitz submanifold.

The generalized differential is the convex hull of gradients of all distance functions "active" at $p$. Geometrically, this corresponds to the convex hull of the normal vectors to all simplex offset surfaces intersecting at $p$.

We do not need to consider a larger number of vectors in the generalized differential, as such points,  if they exist, can be eliminated by a small change in the offset $\epsilon$.
In a general position (i.e., for almost all values of $\epsilon$ and general position of faces of $S$), we can assume that the level-set surfaces of each face intersect transversally,
 except if the point is on a patch shared by two or more such surfaces (e.g., a part of the sphere centered at a common vertex of two triangles), which we exclude from consideration as all offsets share a normal in this case. In other words, two vectors, in the case of intersection of two surfaces and 3 vectors, in the case of three surfaces, are linearly independent.

Note that if two vectors in the generalized differential are linearly dependent and point in opposite directions (tangential contact), then a zero vector is in the convex hull and the matrix $B$ required by the theorem does not exist, but such contact can be removed by a small perturbation of $\epsilon$.

The statement of Clarke's theorem requires finding two vectors,  $v_1$, and $v_2$ (the rows of $B$) such that for any $w \in \partial f$,  $[v_1, v_2, w]$ are linearly independent.    If there is only one vector in $\partial f(p)$ ($C_1$ point) then any two orthogonal vectors would suffice.  If there are two vectors, 
$\partial f$ is a segment of a line in $\mathbb{R}^3$ not passing through zero.   We just need to take any two independent vectors in a plane passing through zero that is parallel to this line.  Then  $(v_1, v_2,  (1-t) w_1 + t w_2)$ are linearly independent for any $0 \leq t \leq1$, as no vector with the endpoint on the line is in the span of $v_1$ and $v_2$.  Similarly, the convex hull of 3 independent vectors $w_i$ in $\partial f(p)$ is a subset of a plane, not containing zero,  so we take the vectors $v_1$ and $v_2$ to be a basis of a plane through zero, parallel to this plane. This proves the existence of $B$ in all cases.
\end{proof}

The infinitesimal offsets are all homeomorphic to each other, which makes their topology unique for all distances smaller than $\hat{\epsilon}$.
To prove this statement, we assume that the simplicial complex $S$ is simplicially embedded (\Cref{def:simpembedding}) in a tetrahedral background mesh $M$. If the input is a simplicial complex, such embedding can always be constructed, possibly with some refinement of $S$, which does not change its geometry \cite{Diazzi2023}. The additional benefit of our proof relying on the background mesh is that it immediately leads to a constructive algorithm for topology-preserving offset approximation.

\begin{definition}[Simplicial Embedding]
A tetrahedral mesh $\M$ is a \emph{simplicial embedding} of a simplicial complex  $S \subset \M$ if for any tetrahedron $t \in \M$ the intersection of  $S$ and $t$ is either empty, or is a vertex, an edge, or a triangle of $S$ and $M$.
\label{def:simpembedding}
\end{definition}

\paragraph{Locality.} 
We now partition the space into a collection of convex cells, enabling us to localize the definition of the offset surface. We subdivide every tetrahedron into four convex regions, using the pattern shown in \Cref{fig:proof-cell}. We denote with $\mathcal{V}_t(v)$ the convex cell corresponding to a vertex $v \in t$. We define a convex cell for each edge $e_{ij}$ between the vertices $v_i$ and $v_j$ by taking the intersection of the convex cells of its vertices $\mathcal{V}_t(e_{ij}) = \mathcal{V}_t(v_i) \cap \mathcal{V}_t(v_j)$, and similarly for every face $\mathcal{V}_t(f_{ijk}) = \mathcal{V}_t(v_i) \cap \mathcal{V}_t(v_j) \cap \mathcal{V}_t(v_k)$ ($f_{ijk}$ has vertices $v_i$, $v_j$, and $v_k$). We denote as $\tau_S(s)$ the open star of the simplex $s$, which contains $s$ and all other simplices in $S$ containing $s$. For example, the star of an edge is the edge itself, plus all the triangles incident to it.

\begin{figure}
    \centering\footnotesize
    \includegraphics[width=0.6\linewidth]{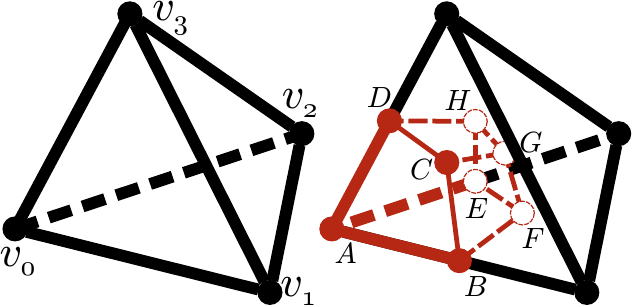}\\
    \caption{Illustration of a tetrahedron with its four vertices (left) and a convex cell in red (right).}
    \label{fig:proof-cell}
\end{figure}

Consider a tetrahedral mesh $M$ that is a simplicial embedding of a simplicial complex $S$. We show that for any choice of $S$, there is a sufficiently small $\epsilon$ for which only $\tau_S(v)$ of a vertex $v$ contributes to the offset in $\mathcal{V}_t(v)$. To do this, we consider the offset we would obtain by considering the maximal surface that could be embedded in $M$, which we denote as $\hat{M}$, composed of all vertices, edges, and triangles in $M$. We prove that for this worst case, the locality holds --- in practice, $S$ would be a subset of $\hat{M}$. This property is crucial, as it allows us to construct the offset locally, unlike traditional offsets whose geometry and topology are global properties.

\begin{lemma}[Offset Locality]
\label{lemma:locality} 
For any tetrahedral mesh $M$ embedded in $\mathbb{R}^3$, there exists an $\hat{\epsilon} > 0$ such that, for every simplex $s \in \hat{M}$ and every tetrahedron $t \in M$ incident to $s$: \\
\begin{equation*}
\big(\mathcal{S}_\epsilon(\hat{M}) \setminus \mathcal{S}_\epsilon(\tau_{\hat{M}}(s)) \big) \cap \mathcal{V}_t(s) = \emptyset
\end{equation*}
\end{lemma}
\begin{proof}
All simplices in $\tau_{\hat{M}}(s)$ contain $s$, and there is always a point $p \in s$ contained in $\mathcal{V}_t(v)$ since $t$ is incident to $s$, and therefore the distance between all simplices in $\mathcal{S}_\epsilon(\tau_{\hat{M}}(s))$ and $\mathcal{V}_t(s)$ is zero. All other simplices in $\hat{M} \setminus \tau_{\hat{M}}(s)$ 
have a positive distance from $\mathcal{V}_t(s)$ as: (1) they cannot intersect the interior of $t$ since $\M$ is embedded and (2) they cannot be in the boundary of $\mathcal{V}_t(s)$ as the intersection of $\mathcal{V}_t(s)$ and $t$ is a subset of $\tau_{\hat{M}}(s)$. For any $\epsilon$ smaller than the minimal distance between $\mathcal{S}_\epsilon(\hat{M}) \setminus \mathcal{S}_\epsilon(\tau_{\hat{M}}(s))$ and $\mathcal{V}_t(s)$ their intersection is empty. If we pick the smallest $\epsilon$ for all $s \in \hat{M}$ and $t \in M$, then the locality condition holds for all convex cells $\mathcal{V}_t(s)$.
\end{proof}

We note that an $\epsilon$ satisfying the locality property for $\hat{M}$ also trivially satisfies the same property for any other embedded simplicial complex $S \subset \hat{M}$.

\begin{lemma}
Let $X$ and $Y$ be connected closed bounded submanifolds of $\mathbb{R}^n$ with boundaries. Let $f: X \rightarrow Y$ be a continuous injective map, bijectively mapping $\partial X$ to $\partial Y$, and its inverse on $f(X)$ is continuous. Then $f$ is a homeomorphism between $X$ and $Y$.
\label{lemma:homeomorphism}
\end{lemma}
\begin{proof}
As f is continuous, its image $f(X)$ in $Y$ is closed and bounded, as $X$ is closed and bounded.  As $f: X \rightarrow f(X)$ is a homeomorphism, $f$ maps $\partial X$ to $\partial f(X)$,  so $\partial f(X) = \partial Y$.
Suppose $f(X)$ does not coincide with the whole $Y$.  As $f(X)$ is closed, $Y \setminus f(X)$  is open, by assumption nonempty, and is contained in the interior of $Y$ ($\mathrm{Int}(Y)$). $\mathrm{Int}(Y)$ is connected because the interior of a connected manifold is connected. Let $y_0$ be an interior point of $f(X)$. Consider a continuous simple path $p : [0,1] \rightarrow Y$ connecting $y_0$  and a point in $\mathrm{Int}(Y) \setminus f(X)$ 
(as $X$ has interior points and $f(\partial X) = \partial Y$, there are interior points in $f(X)$, so $y_0$ exists). 

Consider the subset $Z$ of $[0,1]$ such that $f(Z)$ is contained in $f(X)$  and its complement in $[0,1]$.
As $f(X)$  is closed, then $Z = p^{-1}(f(X) \cap \mathrm{Im}(p))$ is closed, and not coinciding with $\mathrm{Im}(p)$, so it has a boundary at an interior point $q$ of $[0,1]$, for which $p(q)$ is contained in $f(X)$.  As it is a boundary point of $Z$,  any neighborhood of $p(q)$ in $Y$ contains both points of $f(X)$ and $Y \setminus f(X)$, i.e., is on the boundary of $f(X)$.   But by construction of the path, $p(q)$ must be an interior point of $Y$, which contradicts $f(\partial X) = \partial f(X) = \partial Y$.
\end{proof}

\begin{lemma}
There is an $0 < \epsilon \le \hat{\epsilon}$, such that the union of the $\epsilon$-inflations of the elements of $\tau_S(s)$ of a simplex $s \in S$, 
\begin{align*}
\mathcal{S_\epsilon}(\tau_S(s)) = \bigcup_{c \in \tau_S(c)} \mathcal{S_\epsilon}(c),
\end{align*}
is a star domain with $s$ in its kernel and its boundary $\partial \mathcal{S_\epsilon}(\tau_S(s))$ is manifold.
\label{lemma:stardomain}
\end{lemma}
\begin{proof}

Each $\mathcal{S_\epsilon}(c)$ is the Minkowski sum of two convex primitives (a simplex and a sphere), and it is thus convex.
For two star-shaped domains $r_1$ and $r_2$, every point $p$ in the intersection of the two kernels ($p \in \ker(r_1) \cap \ker(r_2)$) has visibility to all points in both $r_1$ and $r_2$ and thus belongs to the kernel of the union of $r_1$ and $r_2$ ($p \in \ker(r_1 \cup r_2)$). 
Since $s$ is a face of all $c_j \in \tau_S(s)$, $s$ is in the kernel of all $\mathcal{S_\epsilon}(c_i)$, and therefore in the kernel of their union. This proves that $\mathcal{S_\epsilon}(\tau_S(s))$ is a star domain with $s$ in its kernel. From \Cref{thm:infinitesimal}, its boundary is manifold as $\tau_S(s)$ is a simplicial complex. 
\end{proof}

Using the result of the previous three lemmas, we will show that a convex cell's intersection with any infinitesimal offset is either empty or a disk.

\begin{figure}
    \centering\footnotesize
    \includegraphics[width=\linewidth]{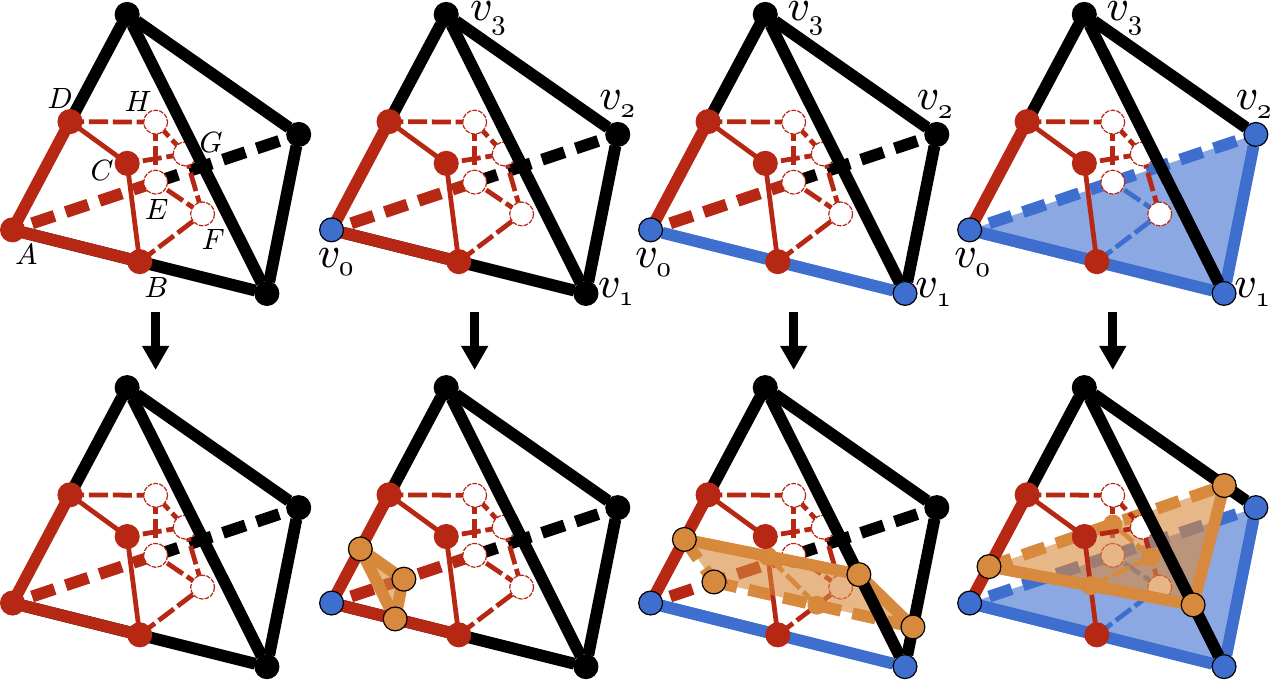}
    \parbox{0.24\linewidth}{\centering (1)}\hfill
    \parbox{0.24\linewidth}{\centering (2)}\hfill
    \parbox{0.24\linewidth}{\centering (3)}\hfill
    \parbox{0.24\linewidth}{\centering (4)}
    \caption{Rules to generate the topological offset (orange) for an embedded simplicial complex (blue). For illustration purposes, we omit the simplicial decomposition of the polyhedra.}
    \label{fig:rules}
\end{figure}

\begin{theorem}[Local Disk Topology]
For any $\epsilon \in (0, \hat{\epsilon})$ and for any cell $\mathcal{V}_t(v_i)$, the intersection of the cell and the offset $\mathcal{S_\epsilon}(S) \cap \mathcal{V}_t(v_i)$ is empty iff $S \cap \mathcal{V}_t(v_i)$ is empty, otherwise it is a topological disk.
\label{thm:dtov}
\end{theorem}

\begin{proof}
We consider the four rules in \Cref{fig:rules} and restrict them to individual convex cells $\mathcal{V}_t(v_i)$ of their vertices. We observe that, after factoring out symmetries, there are only four distinct configurations (Figure \ref{fig:rules}): (1) the convex cell does not contain any simplex of $S$, (2) it contains only a vertex of $S$, (3) it contains a vertex and intersects an edge of $S$, and (4) it contains a vertex and intersects two edges and a triangle of $S$.

\textbf{Case (1)}: Due to locality (\Cref{lemma:locality}), the convex cell does not contain $\partial \mathcal{S}_\epsilon(S)$. 

\textbf{Case (2)}: Due to locality (\Cref{lemma:locality}), the convex cell $\mathcal{V}_t(A)$ contains a part of $\partial \mathcal{S}_\epsilon(S)$ which, inside the cell, is identical to $\partial \mathcal{S}_\epsilon(\tau_S(A))$. 

\emph{Points.} We will show that $\partial \mathcal{S}_\epsilon(\tau_S(A))$ intersects the three edges $AB$, $AE$, and $AD$ at one point per edge. 
Due to locality (\Cref{lemma:locality}), $\epsilon < \|B - A\|$, as the offset of $A$ would otherwise contribute to another convex cell. As $\mathcal{S}_\epsilon(\tau_S(A))$ is a star domain with $A$ in its kernel ($\epsilon > 0$) and $B$ outside ($\epsilon < \|B - A\|$), the edge $AB$ must intersect $\partial \mathcal{S}_\epsilon(\tau_S(A))$ exactly once. The same holds for $AE$ and $AD$.
From locality, it also follows that $\partial \mathcal{S}_\epsilon(\tau_S(A))$ does not intersect any edge of $\mathcal{V}_t(A)$ that is not incident to $A$, as they are also part of other convex cells. 

\emph{Curves.} 
We will show that $\partial \mathcal{S}_\epsilon(\tau_S(A))$ intersects the three quadrilateral faces $ABCD$, $ABFE$, and $AEHD$ in three simple curves, one for each face, whose boundary is composed of the vertices identified before. 
Due to locality, the sphere $\partial \mathcal{S}_\epsilon(A)$ intersects $AB$ and $AD$, but not $BC$ or $CD$. Therefore, the intersection $\partial \mathcal{S}_\epsilon(A) \cap ABCD = \gamma$ is a circular segment.
Let $r$ be a ray from $A$ to a point on $\gamma$. $r$ intersects $\partial S_\epsilon(\tau_S(A)) \cap ABCD$ in a unique point as $S_\epsilon(\tau_S(A)) \cap ABCD$ is star-shaped with $A$ in its kernel (\Cref{lemma:stardomain}), and all intersection points are contained in the interior of $ABCD$ due to \Cref{lemma:locality}.
Therefore, $r$ establishes an injective mapping between $\gamma$ and $\partial S_\epsilon(\tau_S(A)) \cap ABCD$. Due to \Cref{lemma:homeomorphism}, this map is bijective.
As $\gamma$ is a simple curve, it follows that $\partial S_\epsilon(\tau_S(A)) \cap ABCD$ is also a simple curve. The same argument can be applied to $ABFE$ and $AEHD$.

\emph{Patch.}
The intersection $\Gamma = \partial \mathcal{S}_\epsilon(A) \cap \mathcal{V}_t(A)$ is a surface with disk topology (a spherical, simply connected patch) due to locality.
Its boundary $\partial \Gamma$ maps via ray casting, with $A$ as origin, to the points and curves identified before. Ray casting with $A$ as origin defines an injective map from the interior of $\Gamma$ to $\Gamma_{\tau S} = \partial \mathcal{S}_\epsilon(\tau_S(A)) \cap \mathcal{V}_t(A)$, as $\Gamma_{\tau S}$ is a star-shaped domain with $A$ in its kernel, and $\partial \mathcal{S}_\epsilon(\tau_S(A))$ must not intersect $BFGC$, $FGHE$, or $CGHD$ (locality).

Both $\Gamma$ and $\Gamma_{\tau S}$ are connected closed 2-manifold with metric inherited from $\mathbb{R}^3$ with a single boundary (the intersection of both with the boundary of $V_t$ is a 1-manifold as we show above).
Both are connected: if it consisted of several connected components with boundary, then we would have several separate boundary loops.
So only one connected component can have a boundary; the rest need to be without boundary. But then these connected components would also be separate components of the entire offset surface, as there cannot be a continuous path connecting them to any point of the offset outside $\mathcal{V}_t$, not intersecting $\mathcal{V}_t$ boundary.
From \Cref{lemma:homeomorphism}, it follows that the projection is a homeomorphism, and thus $\Gamma_{\tau S}$ is a disk.  

\textbf{Case (3)}: 
\emph{Points.} $\partial S_\epsilon(\tau_S(A))$ 
intersects the 4 edges $AE$, $AD$, $BF$, and $BC$ at one point per edge, and it does not intersect any of the other edges of $\mathcal{V}_t(A)$. The proof is the same as for Case (2) for $AE$ and $AD$. For $BF$, we observe that $\partial \mathcal{S}_\epsilon(\tau_S(A)) \cap BF = $ $\partial \mathcal{S}_\epsilon(\tau_S(e_{01})) \cap BF$ (\Cref{lemma:locality}) and 
$\partial \mathcal{S}_\epsilon(\tau_S(e_{01}))$ is a star domain with $B$ in its kernel ($\epsilon > 0$) and $F$ outside ($\epsilon < \|F - B\|$), the edge $BF$ must intersect $\partial \mathcal{S}_\epsilon(\tau_S(A))$ exactly once. The proof for $BC$ is similar.

\emph{Curves.} $\partial S_\epsilon(\tau(A))$ intersects the 4 quadrilateral faces $ABCD$, $ABFE$, $BFGC$, and $AEHD$ in four curves, one for each face, whose boundary is composed of the vertices identified before. 
The intersection of $\partial \mathcal{S}_\epsilon(B) \cap BFGC$ is a circular segment, thus a simple curve. We establish an injective mapping between $\partial \mathcal{S}_\epsilon(B) \cap BFGC$ and $\partial \mathcal{S}_\epsilon(\tau_S(A)) \cap BFGC$ via ray casting as in Case (2), using $B$  as the origin and noting that $B$ is in the kernel of $\partial \mathcal{S}_\epsilon(\tau_S(e_{01}))$ (\Cref{lemma:stardomain}) and $\partial \mathcal{S}_\epsilon(\tau_S(e_{01})) \cap BFGC = \partial \mathcal{S}_\epsilon(\tau_S(A)) \cap BFGC$. This map is bijective due to \Cref{lemma:homeomorphism}. The proof for $AEHD$ is the same as for Case (2), using $A$ as the ray origin.

For $ABCD$ and $ABFE$, the curves of interest intersect two opposite sides of the quads, so the polar argument is slightly different. In both cases, we need to consider $A$ as the projection point: the bijection is between a part of a circle cut out by the sides of the polygon and the segment of the boundary of 
$\cS_\epsilon(S) \cap V_t(A)$. Consider points $C'$ and $D'$ -- intersection points on edges $BC$ and $AD$.  On a small circle around $A$, these project to $C''$ and $D''$. $C''$ is in the interior of $ABCD$, not on an edge. Any ray through $A$ in the interior of $ABCD$ and to the same side from $AC''$ as 
$D''$, passes through exactly one point of the circle and one point of  $\partial \cS_\epsilon(S) \cap V_t(A)$. It remains to prove that no ray on the other side of $AC''$ intersects  $\partial \cS_\epsilon(S) \cap V_t(A)$. Suppose it does: then if we extend it further, it will intersect the segment $BC$ on the side closer to $B$, i.e. in the interior of $\cS_\epsilon(S) \cap V_t(A)$, this means that it has to intersect $\partial \cS_\epsilon(S) \cap V_t(A)$  twice on its way from $A$ to $BC''$ which contradicts the fact that it is star-shaped. We conclude that there is a bijection.

\emph{Patch.}
Same proof as Case (2).

\textbf{Case (4)}: 
\emph{Points.} $\partial \mathcal{S}_\epsilon(\tau_S(A))$ intersects the 4 edges $AD$, $BC$, $FG$, and $EH$ at one point per edge, and it does not intersect any of the other edges of $\mathcal{V}_t(A)$. The proof is the same as Case (2) for $AD$ with A in the kernel, same as Case (3) for $BC$ and $EH$, with $B$ and $E$ in the kernel, respectively. For $FG$, we observe that $\partial \mathcal{S}_\epsilon(\tau_S(A)) \cap FG = \partial \mathcal{S}_\epsilon(\tau_S(f_{012})) \cap BF$ (\Cref{lemma:locality}) and 
$\partial \mathcal{S}_\epsilon(\tau_S(f_{012}))$ is a star domain with $F$ in its kernel ($\epsilon > 0$) and $G$ outside ($\epsilon < \|G - G\|$), the edge $FG$ must intersect $\partial \mathcal{S}_\epsilon(\tau_S(A))$ exactly once.

\emph{Curves.} $\partial \mathcal{S}_\epsilon(\tau(A))$ intersects 4 quadrilateral faces $ABCD$, $BFGC$, $FGHE$, and $AEHD$, in four simple curves, one for each face, whose boundary is composed of the vertices identified before. The proof for all 4 faces is similar to face $ABCD$ of Case (3), using $A$, $B$, $E$, and $A$ as projection points. 

\emph{Patch.}
Same proof as Case (2). 

\end{proof}

From \Cref{thm:dtov}, it follows that all infinitesimal offsets are homeomorphic to each other, as a bijective map can be defined between each disk inside each convex cell.

\paragraph{Topological Offset}

We established that infinitesimal offsets of a simplicial complex are all manifold surfaces homeomorphic to each other: their topology is unique for a given input surface.
To decouple geometry and topology, we introduce the concept of a \emph{topological offset} (Figure \ref{fig:smooth}, green curve), which is a surface that is homeomorphic to an infinitesimal offset, but whose geometric embedding is arbitrary (as long as it does not intersect the input surface). The reason for not using the infinitesimal offset directly is practical: representing it in a mesh might lead to infinitesimally small elements: on the other hand, topological offsets can be computed reliably and robustly.

\begin{definition}[Topological Offset]
A simple surface $O$ is a \emph{topological offset} of a 
non-intersecting simplicial complex $S \subset \Omega$ if  $O$ and $S$ do not intersect and $O$ is homeomorphic to an infinitesimal offset of $S$, i.e., there exists a continuous and bijective function $\hat c \colon O \rightarrow \mathcal{S_\epsilon}(S)$ for any $0 < \epsilon < \hat{\epsilon}$.

\label{def:topoffset}
\end{definition}
The function $\hat c$ is a topological version of the closest point function $c$ used in the offset definition.

\section{Topological Offset Construction}
\label{sec:discrete_topological_offset}

\begin{figure}
    \centering\footnotesize
    \includegraphics[width=0.96\linewidth]{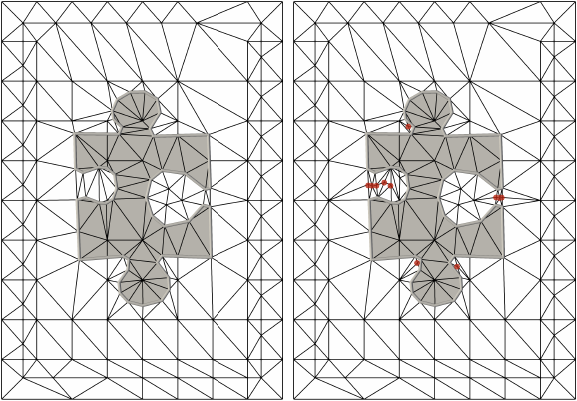}\\
    \parbox{0.48\linewidth}{\centering (1) Input}
    \parbox{0.48\linewidth}{\centering (2) Simplicial Embedding}
    \includegraphics[width=0.96\linewidth]{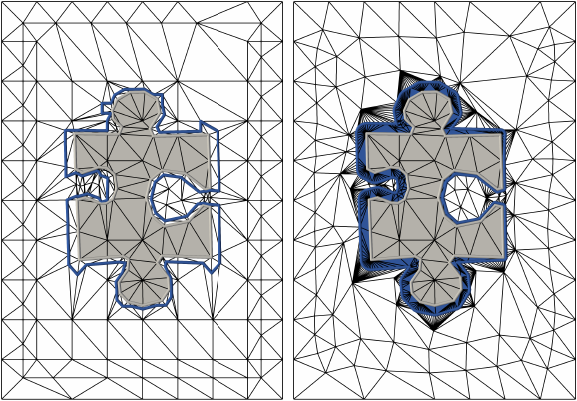}\\
    \parbox{0.48\linewidth}{\centering (3) Offset Initialization}
    \parbox{0.48\linewidth}{\centering (4) Optimization}
    \caption{The full 2D pipeline for generating the offset: (1) we start with a
    simplicial complex embedded in a background mesh, (2) we make the embedding simplicial (see highlighted vertices), which allows us to (3) generate an initial topological offset, and (4) we optionally optimize the offset while keeping the background mesh and offset valid.}
    \label{fig:flowchart}
\end{figure}

We now describe an algorithm (\Cref{fig:flowchart}) to compute a topological offset of a simplicial complex $S$. The algorithm has two steps (\Cref{sec:simplicial_embedding} and \Cref{sec:insertion}) that are parameter-free, followed by an optional geometrical optimization to improve the offset quality and embedding, which we will describe in Section \ref{sec:optimization}. 

\paragraph{Input Preprocessing}
If the input simplicial complex is not embedded in a tetrahedral mesh, we embed it using TetWild \cite{hu2018tetrahedral}. Note that other algorithms could be used for this step, such as TetGen \cite{tetgen} or Robust CDT \cite{Diazzi2023}. The embedded complex is a subset of the triangles, edges, or vertices in the tetrahedral mesh tagged as belonging to the complex.

\begin{figure}
    \centering\footnotesize
    \includegraphics[width=0.32\linewidth]{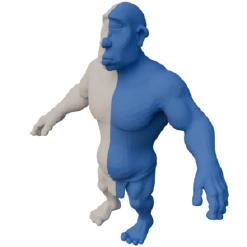}\hfill
    \includegraphics[width=0.32\linewidth]{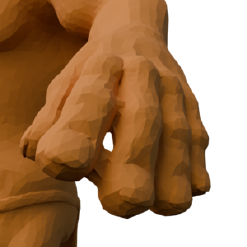}\hfill
    \includegraphics[width=0.32\linewidth]{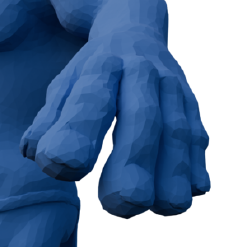}
    \parbox{0.32\linewidth}{\centering input}\hfill
    \parbox{0.32\linewidth}{\centering no simplicial embedding}\hfill
    \parbox{0.32\linewidth}{\centering simplicial embedding}
    \caption{By making the background mesh a simplicial embedding of the input (white), we guarantee that the topological offset (blue) is homeomorphic to the infinitesimal offset family. Without simplicial embedding, the topology might be different (orange).}
    \label{fig:why-simplicial-embedding}
\end{figure}

\paragraph{Step 1: Simplicial Embedding}
We create a simplicial embedding (\Cref{sec:simplicial_embedding}) using a sequence of local operations. The pseudocode for this step is described in Algorithm \ref{alg:simplicial_embedding}. This condition is essential to ensure the existence of the discrete counterpart of a topological offset (\Cref{fig:why-simplicial-embedding}).

\paragraph{Step 2: Offset Insertion}
We insert a topological offset in the background mesh using a topological variant of marching tetrahedra (\Cref{sec:insertion}). 

\paragraph{Step 3: Offset Optimization}
Optionally, we optimize the tetrahedral mesh (and consequently the offset) by increasing its quality and adapting for various applications, as described in \Cref{sec:optimization}.

\paragraph{Output Postprocessing}
The output of our algorithm is a tetrahedral mesh with an embedded offset surface. The background mesh can be discarded if the downstream application only needs the surface mesh.

\subsection{Step 1. Simplicial Embedding}
\label{sec:simplicial_embedding}

\begin{figure}
    \centering
    \includegraphics[width=.5\linewidth]{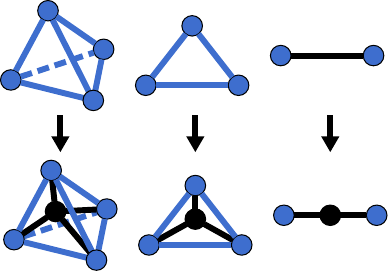}
    \caption{Split operation for a tetrahedron, triangle, and edge.}
    \label{fig:splitsimplex}
\end{figure}

We propose an algorithm to convert a simplicial complex embedded in a tetrahedral background mesh $M$ (such as the one generated by \cite{hu2018tetrahedral}) into a simplicial embedding by performing a sequence of edge splits  (\Cref{fig:splitsimplex}) to update $\M$. 
\Cref{alg:simplicial_embedding} iterates over every tetrahedron $t\in \M$, and checks if the boundary of $t$ is in $S$ (\cref{l:tet}). In that case, it splits the tetrahedron (\cref{l:split}). It then iterates the same procedure on all triangles (\cref{l:triangle}) and edges (\cref{l:edge}) that are not in $S$ (\Cref{fig:embedding_examples} shows some examples).
Note that a split to an edge or a face is propagated to the tetrahedra intersecting it.

\begin{algorithm}
\caption{\Call{simplicial\_embedding}{$\M, S$}:}\label{alg:cap}
\begin{algorithmic}[1]
\Require A mesh $\M$ with an embedded simplicial complex $S$
\Ensure $\M'$ is a simplicial embedding of $S$
\State $M' \gets M$
\For{each tetrahedron $t \in M'$} \label{l:tet}
\If{$t \cap S$ contains four triangles} 
    \State $\M' \gets$ \Call{split\_tetrahedron}{$\M',t,S$} \label{l:split}
\EndIf
\EndFor
\For{each triangle $t \in M'$} \label{l:triangle}
\If{$t \not\in S$ and $t \cap S$ contains three edges} 
    \State $\M' \gets$ \Call{split\_triangle}{$\M',t,S$}
\EndIf
\EndFor
\For{each edge $e \in M'$} \label{l:edge}
\If{$e \not\in S$ and $e \cap S$ contains two vertices} 
    \State $\M' \gets$ \Call{split\_edge}{$\M',e,S$}
\EndIf
\EndFor
\State \Return $M'$
\end{algorithmic}
\label{alg:simplicial_embedding}
\end{algorithm}

\begin{figure}
    \centering
    \includegraphics[width=.55\columnwidth]{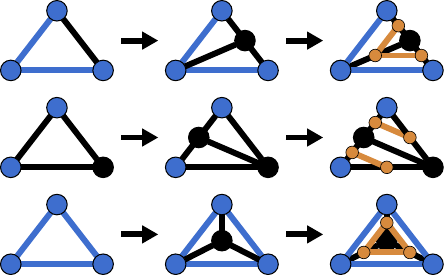}
    \caption{2D exemplary results for applying \Cref{alg:simplicial_embedding} to make $\M$ a simplicial embedding of $S$ (blue), and the resulting topological offsets (orange). For illustration purposes, we omit the simplicial decomposition of the polygons.}
    \label{fig:embedding_examples}
\end{figure}

\begin{theorem}[Simplicial Embedding of a Simplicial Complex]
\Cref{alg:simplicial_embedding} produces a mesh $\M$ such that $\M$ is a simplicial embedding of $S$.
\end{theorem}

\begin{proof}
After the first loop of \Cref{alg:simplicial_embedding} completes, no tetrahedron has four faces in $S$. If a tetrahedron has three or two faces in $S$, then there is a face or faces that have three edges in $S$, and these faces will be split in the second loop.  Therefore, after the first two loops, all tetrahedra have no more than one face in $S$.   
Suppose a tetrahedron has a face and two or more edges not contained in this face in $S$. 
Then there are triangles not in $S$ with three edges in $S$ that will be split in the second loop, i.e., there is no more than one edge in $S$ after the second loop. 
Thus, after the first two loops, a tetrahedron may contain at most one face of $S$ and at most one edge not contained in this face, and some number of vertices of $S$.  
Similarly, in the last loop, if there is a face and an edge, or a face and a vertex, or an edge and a vertex, or two vertices,  there will be always an edge not in $S$ with two vertices in $S$ that will be split in the third loop.  
  We conclude that the output is a simplicial embedding of $S$. 

\end{proof}

\begin{figure}
    \centering\footnotesize
    \includegraphics[width=0.9\linewidth]{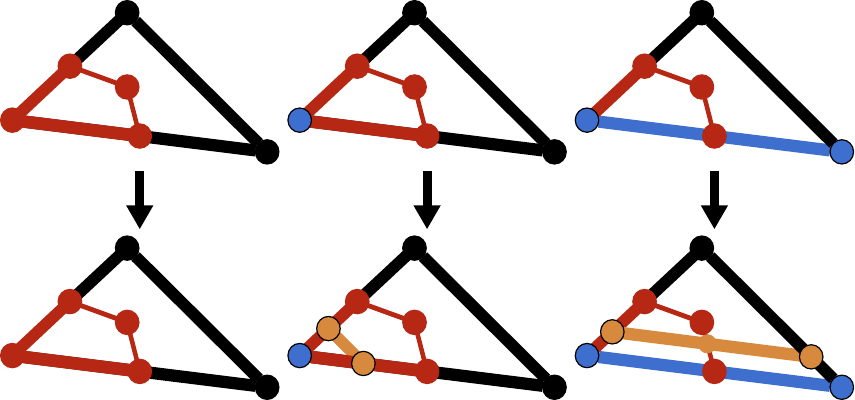}\\
    \parbox{0.3\linewidth}{\centering (1)}
    \parbox{0.3\linewidth}{\centering (2)}
    \parbox{0.3\linewidth}{\centering (3)}
    \caption{Rules to generate the topological offset in 2D. For illustration purposes, we omit the simplicial decomposition of the polygons. A convex cell as it is used in \Cref{thm:dtov} is depicted in red.}
    \label{fig:rules_2d}
\end{figure}

\subsection{Step 2. Offset Insertion}
\label{sec:insertion}
We insert the discrete topological offset into the background mesh by using a binary version of marching tetrahedra \cite{gueziec1995exploiting}: every tetrahedron that contains a vertex, an edge, or a triangle of the input simplicial complex is partitioned using the rules in Figure \ref{fig:rules}. We do not perform an interpolation as suggested in \cite{gueziec1995exploiting}, but we just place the inserted vertices at the midpoint of the split edges.

\begin{theorem}
The application of the rules in Figure \ref{fig:rules} to a simplicial complex $S$ simplicially embedded in a background mesh $M$ is a discrete topological offset homeomorphic to a continuous one.
\label{thm:homeomorphic}
\end{theorem}
\begin{proof}

\textbf{Case (1)}:
No elements are added in this case.

\textbf{Case (2)}:
$\mathcal{S}_\epsilon(S) \cap \mathcal{V}_t(v_0)$ is homeomorphic to the orange triangle in \Cref{fig:rules} case (2) because it is also a disk intersecting the same simplices. All other cells in the tetrahedron $t$ do not intersect $\mathcal{S}_\epsilon(S)$.
Thus, the rule constructs a disk homeomorphic to $\mathcal{S}_\epsilon(S)$ within $t$.

\textbf{Case (3)}:
The cells $\mathcal{V}_t(v_0)$ and $\mathcal{V}_t(v_1)$ have a non-empty intersection with $\mathcal{S}_\epsilon(S)$, and the offset $\partial\mathcal{S}_\epsilon(S)$ within every cell of the tetrahedron $t$ is a disk (\Cref{thm:dtov}). The cells share the face $BFGC = \mathcal{V}_t(e_{01})$ and the disks' boundaries intersect that face in the curve $\partial\mathcal{S}_\epsilon(S)\cap \mathcal{V}_t(e_{01})$. It follows that the union of the two disks is also a disk. The orange polygon (that can be decomposed into two triangles) in \Cref{fig:rules} case (3) is also a disk and thus homeomorphic to $\mathcal{S}_\epsilon(S)$ within $t$.

\textbf{Case (4)}:
The cells $\mathcal{V}_t(v_0)$, $\mathcal{V}_t(v_1)$, and $\mathcal{V}_t(v_2)$ have a non-empty intersection with $\mathcal{S}_\epsilon(S)$, and the offset $\partial\mathcal{S}_\epsilon(S)$ within every cell of the tetrahedron $t$ is a disk (\Cref{thm:dtov}). The cells share faces in the same manner as in case (2), e.g., $\mathcal{V}_t(v_0)$ and $\mathcal{V}_t(v_2)$ share the face $EFGH = \mathcal{V}_t(e_{02})$. It follows that the union of the three disks is also a disk. The orange triangle in \Cref{fig:rules} case (4) is also a disk and thus homeomorphic to $\mathcal{S}_\epsilon(S)$ within $t$.

\end{proof}

The mesh, composed of all polygons constructed from the rules in \Cref{fig:rules}, is manifold, as it is homeomorphic to an infinitesimal offset (\Cref{thm:homeomorphic}) which is manifold (\Cref{thm:infinitesimal}).

While we prove the correctness of our method only for a simplicial complex embedded in a tetrahedral mesh, it is possible to prove a similar result in 2D, which leads to the patterns in \Cref{fig:rules_2d}. We show an example of the 2D algorithm in \Cref{fig:hello-world}, where the input is a sequence of marked edges in a triangle mesh. 
\begin{figure}
    \centering
    \includegraphics[width=0.9\linewidth]{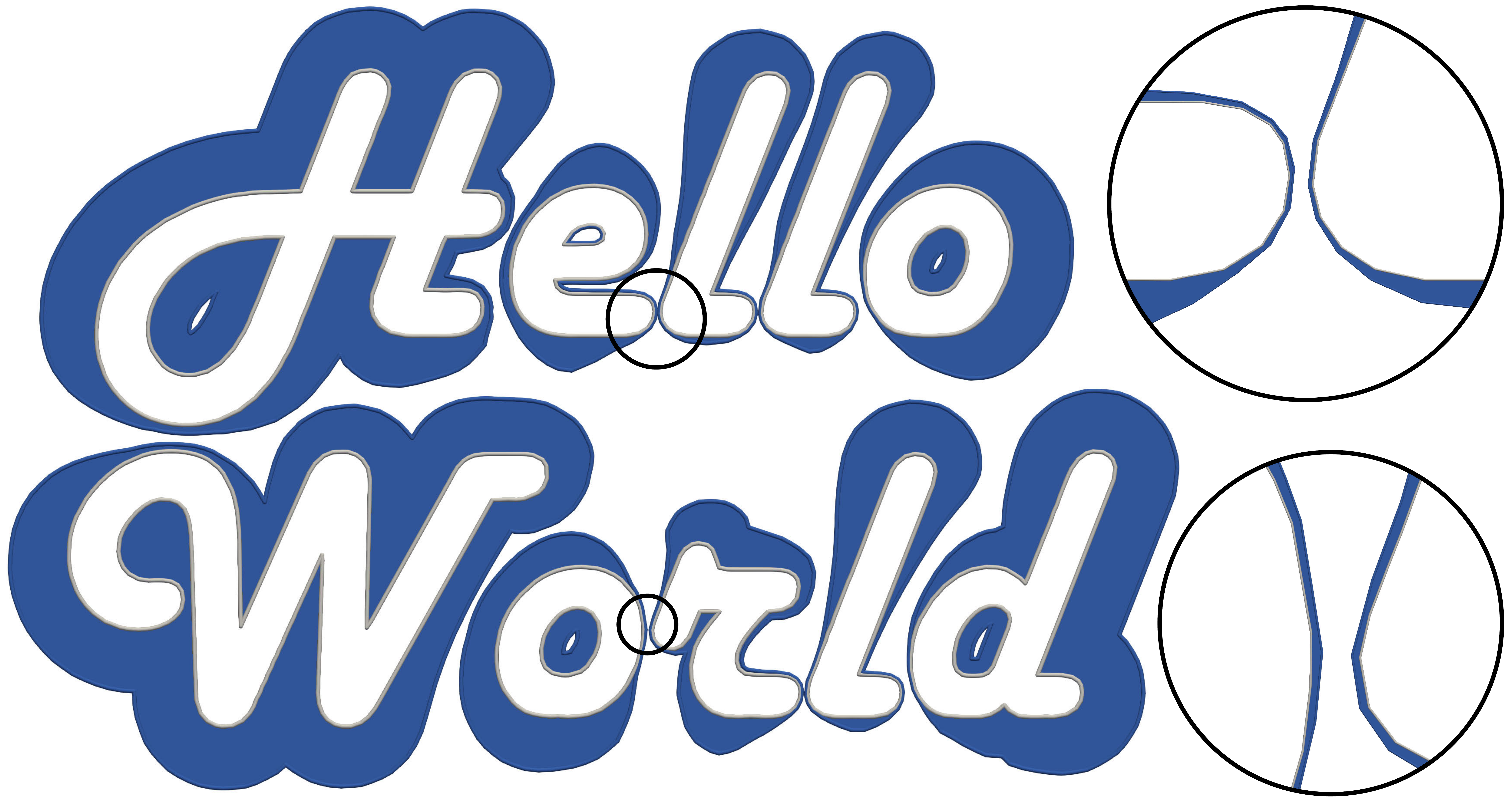}
    \caption{Our method can also construct topological offsets in 2D.}
    \label{fig:hello-world}
\end{figure}

We note that \Cref{alg:simplicial_embedding} is purely topological. The coordinates of the vertices of $\M$ are never used, which makes our algorithm unconditionally robust. We note, however, that the computation of the vertex positions after a split might lead to inverted elements if the rounding error is larger than the length of the edge. We show an example of such a problematic case in \Cref{fig:failure-case}, and note that this only happened once in our large-scale evaluations (\Cref{sec:results}). See \Cref{sec:robustness} for additional details on this problem.

\subsection{Step 3. Offset Optimization}
\label{sec:optimization}

The inserted topological offset depends on the resolution of the background mesh: a finer mesh will create an offset closer to the input surface while a coarser mesh will have an offset further away. To remove this dependence from the background mesh and to allow customization of the embedded offset, we introduce an optimization algorithm that modifies the connectivity and geometry of the offset while preserving its topology and avoiding self-intersections and intersections with the input. 

\begin{figure}
    \centering
    \includegraphics[width=0.3\linewidth]{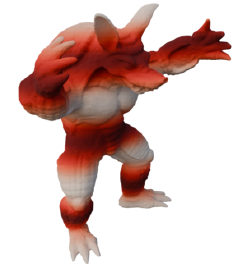}\hfill
    \includegraphics[width=0.3\linewidth]{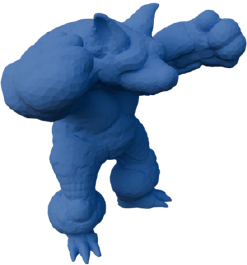}\hfill
    \includegraphics[width=0.3\linewidth]{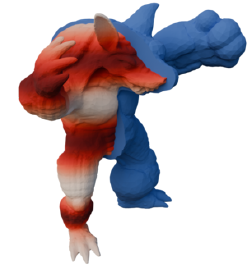}
    \caption{A varying offset distance applied to the armadillo.}
    \label{fig:spatially_varying_offset}
\end{figure}

\paragraph{User Parameters}
While the generation of a topological offset does not require any parameters, 
we introduce a primary parameter to control desired distance $\delta$ of the offset from the input simplicial complex (which could be spatially adaptive, \Cref{fig:spatially_varying_offset}). 
We note that this is a desideratum: it is, in general, impossible to guarantee an arbitrary distance while preserving the infinitesimal offset topology (\Cref{fig:knot_adaptation}). 
\begin{figure}
    \centering\footnotesize
    \includegraphics[width=0.19\linewidth]{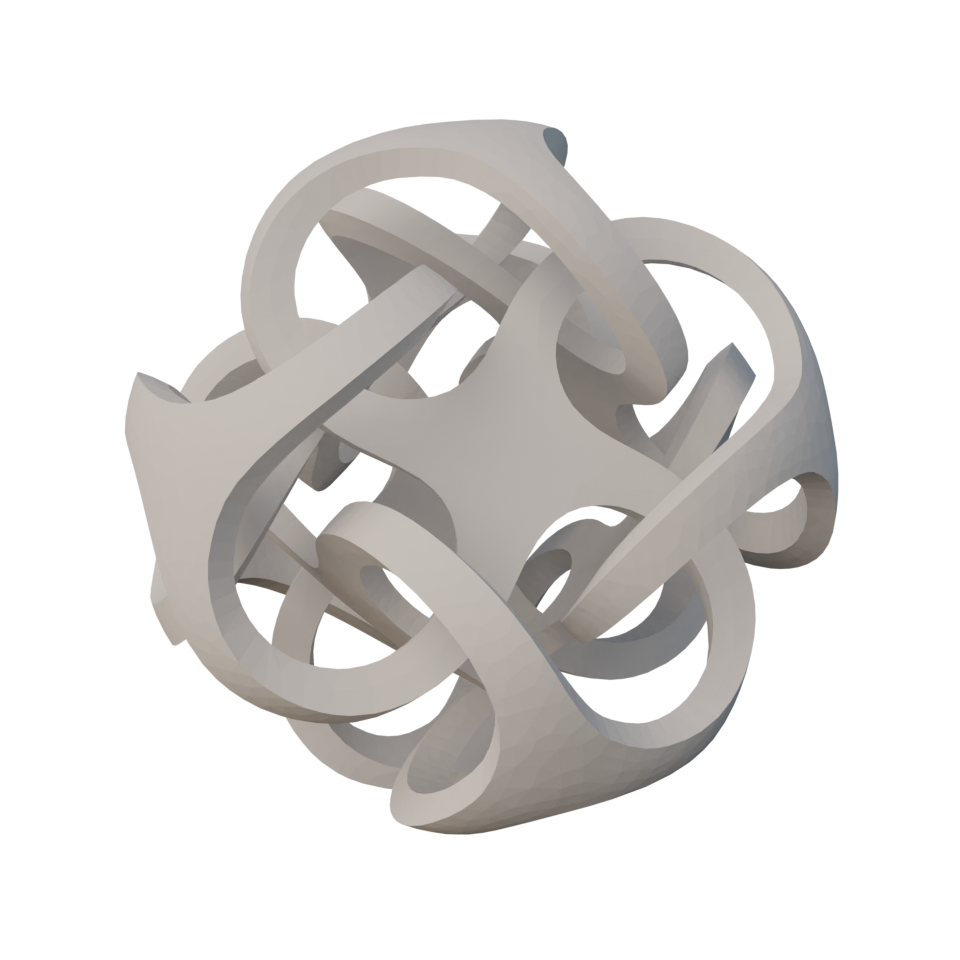}\hfill
    \includegraphics[width=0.19\linewidth]{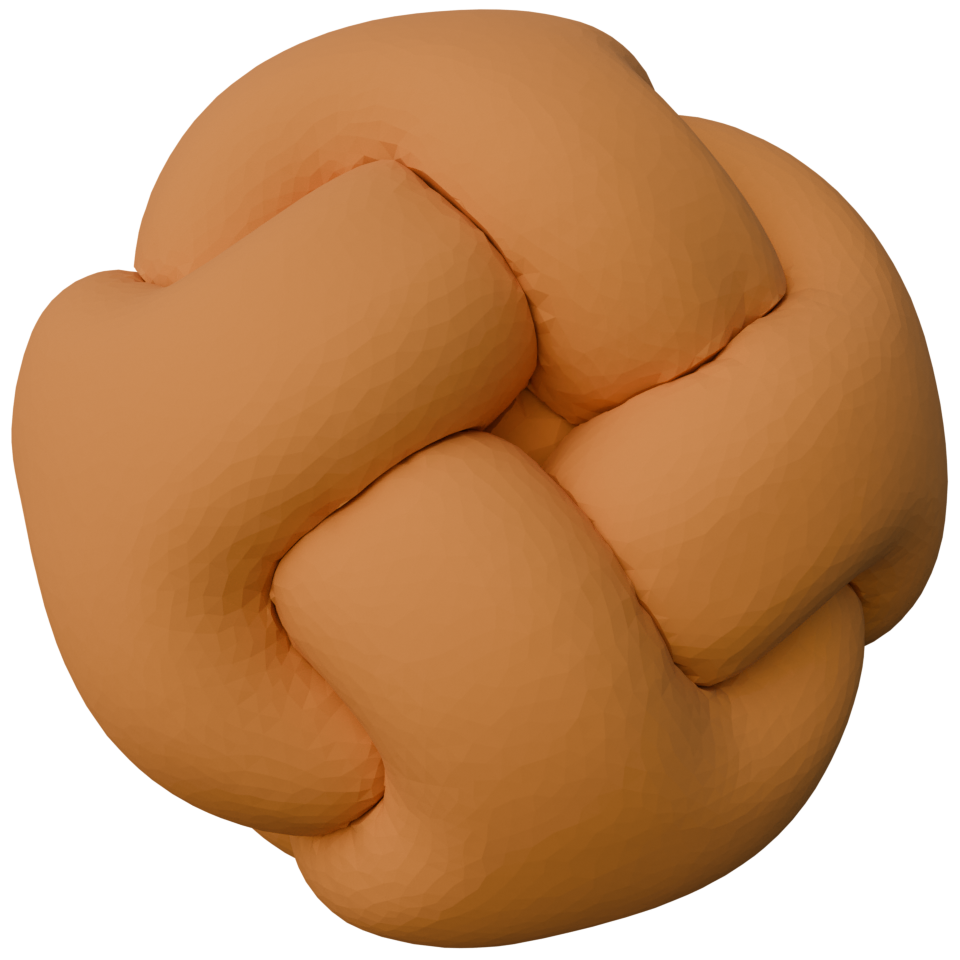}\hfill
    \includegraphics[width=0.19\linewidth]{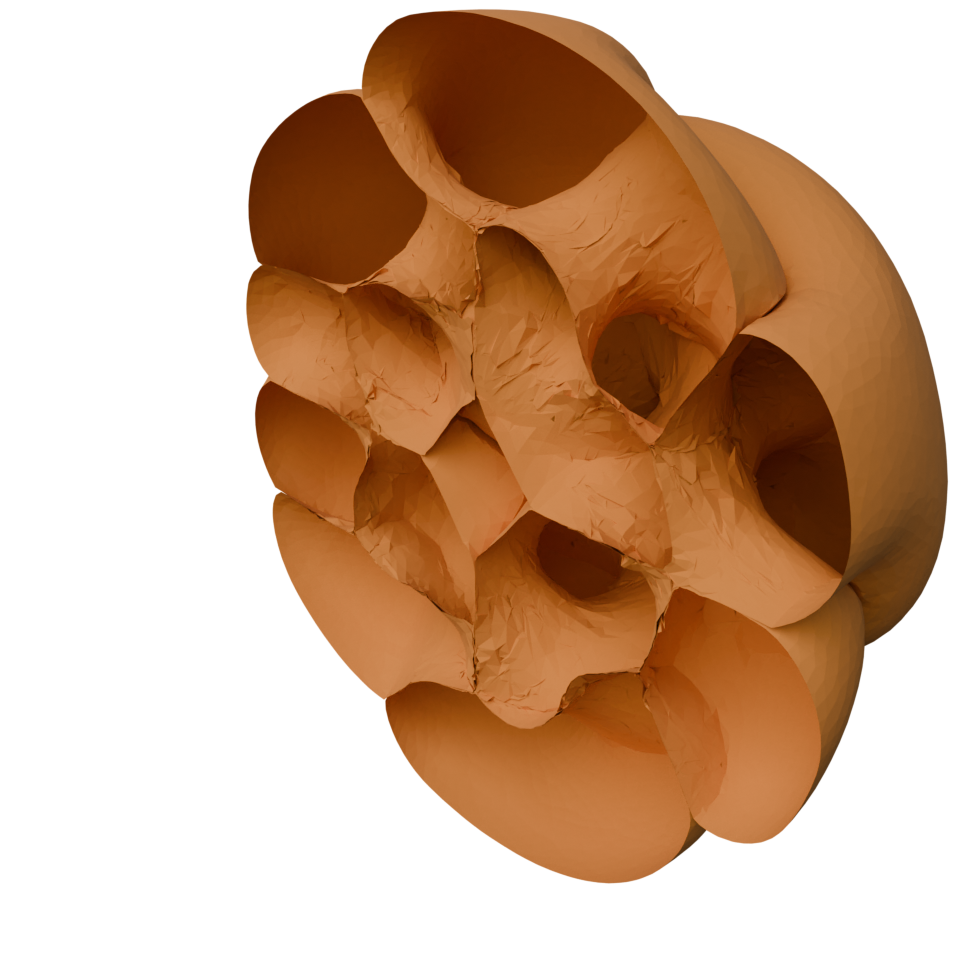}\hfill
    \includegraphics[width=0.19\linewidth]{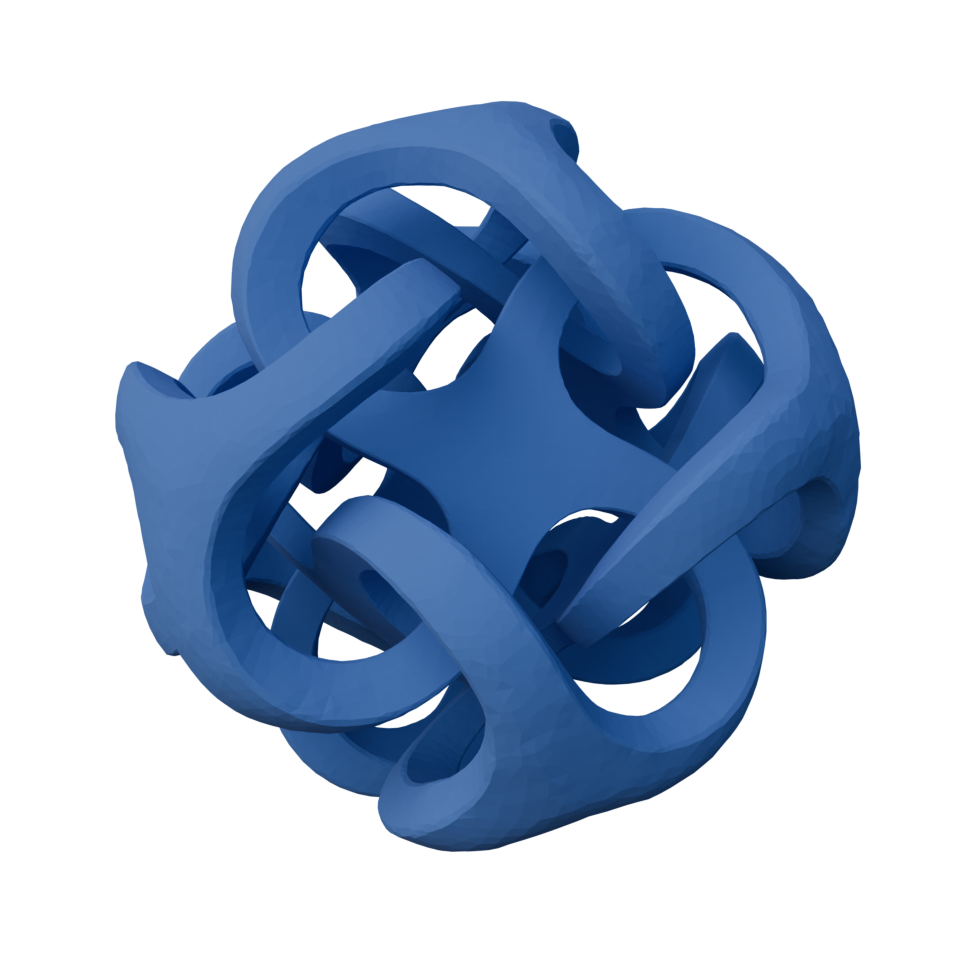}\hfill
    \includegraphics[width=0.19\linewidth]{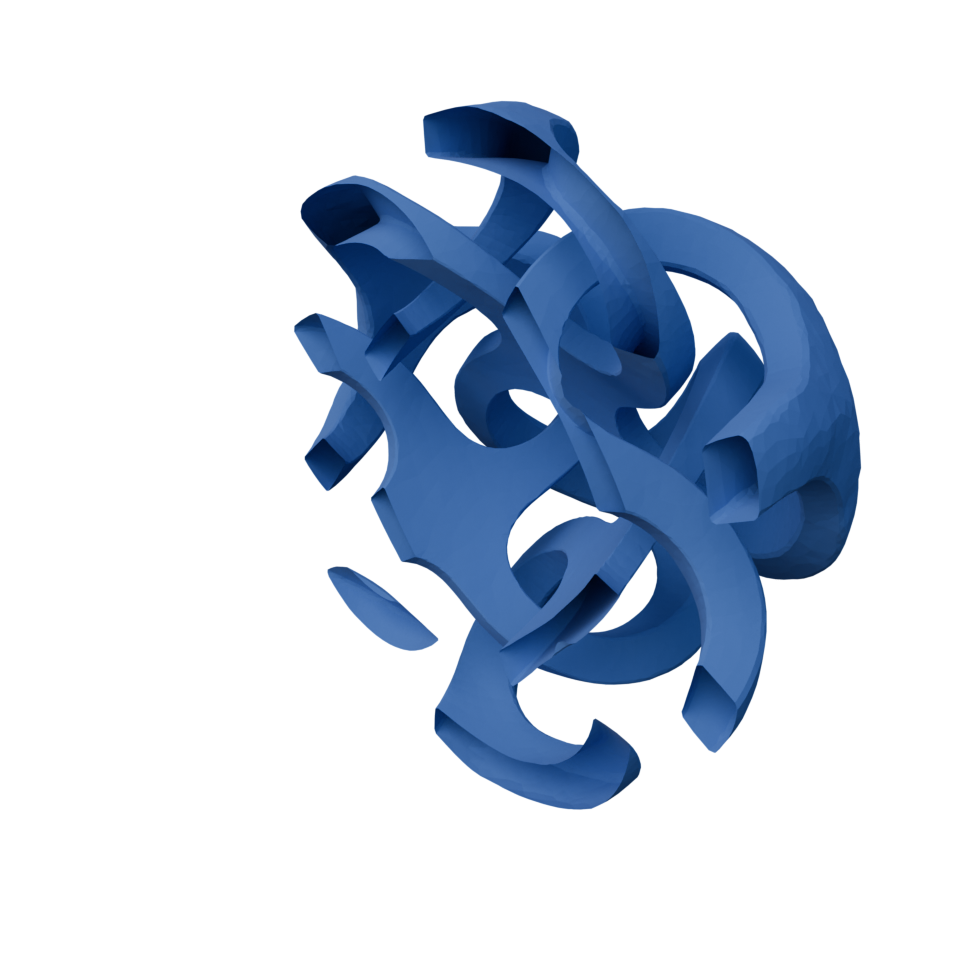}
    \parbox{0.19\linewidth}{\centering input}\hfill
    \parbox{0.38\linewidth}{\centering uniform offset distance}\hfill
    \parbox{0.38\linewidth}{\centering adapted offset distance}
    \caption{We prevent offset from getting too close (orange) by locally adapting the offset distance (blue).}
    \label{fig:knot_adaptation}
\end{figure}
The other parameters control the tradeoff between geometrical accuracy and running time: (1) the maximum normal deviation $\sigma_\text{max}=15^\circ$ which controls how well the mesh should adapt to the offset curvature (\Cref{subsec:quality_metrics}), (2) $\sigma_\text{min}=2^\circ$ which controls when the offset curvature is considered planar, (3) $\ell_\text{max} = \infty$ and (4) $\ell_\text{min} = 2 \delta \sin(\sigma_\text{max})$ the maximum/minimum edge length. 

Increasing the maximum normal deviation leads to a less smooth offset surface (\Cref{fig:parameters}), while increasing $\ell_\text{min}$ leads to coarse results (\Cref{fig:coarse_offset}). Finally, our method can generate \emph{outside} and \emph{inside}
offsets if the input is closed and oriented
(\Cref{fig:inside-outside}). 

\begin{figure}
    \centering
    \includegraphics[width=\linewidth]{fig/parameter_study/params.pdf}
    \caption{Effect of the offset distance and normal deviation on the output.}
    \label{fig:parameters}
\end{figure}

\begin{figure}
    \centering\footnotesize
    \parbox{0.24\linewidth}{\centering
        \includegraphics[width=\linewidth]{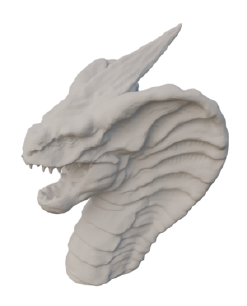}
        input\\ $\#t = 48\,126$
    }\hfill
    \parbox{0.24\linewidth}{\centering
        \includegraphics[width=\linewidth]{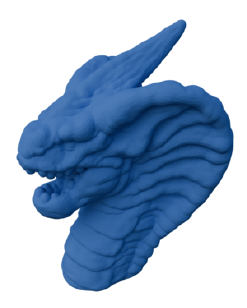}
        $\delta = 1\%$\\ $\#t = 100\,442$
        \includegraphics[width=\linewidth]{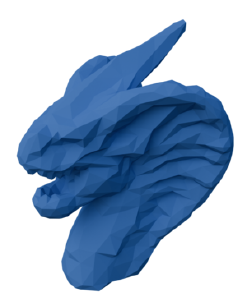}
        $\delta = 1\%$\\ $\#t = 2\,518$
    }\hfill
    \parbox{0.24\linewidth}{\centering
        \includegraphics[width=\linewidth]{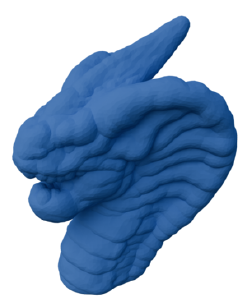}
        $\delta = 2\%$\\ $\#t = 32\,082$
        \includegraphics[width=\linewidth]{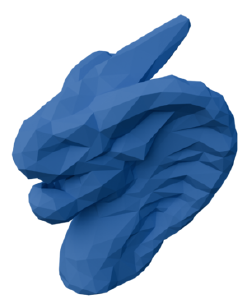}
        $\delta = 2\%$\\ $\#t = 1\,554$
    }
    \hfill
    \parbox{0.24\linewidth}{\centering
        \includegraphics[width=\linewidth]{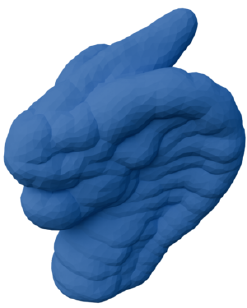}
        $\delta = 4\%$\\ $\#t = 10\,298$
        \includegraphics[width=\linewidth]{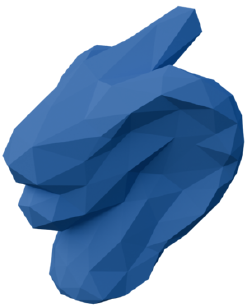}
        $\delta = 4\%$\\ $\#t = 458$
    }
    \caption{Setting $\ell_\text{min} = 100\%$ (relative to the bounding box diagonal) and $\sigma_\text{max} = 90^\circ$ leads to coarse topological offsets. The top row shows the offset meshes with default parameters.}
    \label{fig:coarse_offset}
\end{figure}

\begin{figure}
    \centering\footnotesize
    \includegraphics[width=0.32\linewidth]{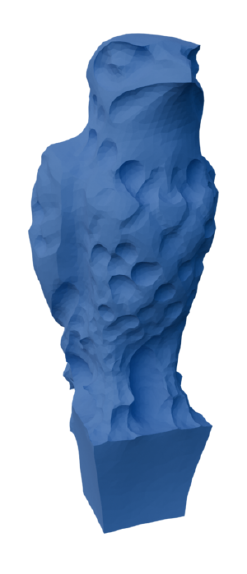}\hfill
    \includegraphics[width=0.32\linewidth]{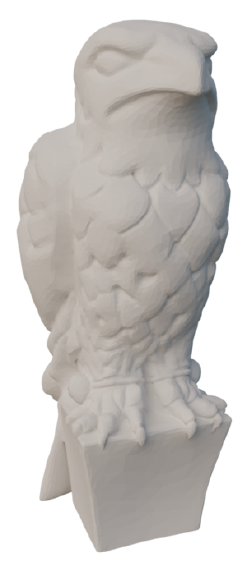}\hfill
    \includegraphics[width=0.32\linewidth]{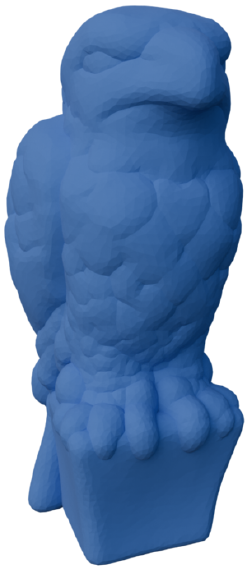}
    \parbox{0.32\linewidth}{\centering $-2\,\%$}\hfill
    \parbox{0.32\linewidth}{\centering input}\hfill
    \parbox{0.32\linewidth}{\centering $+2\,\%$}
    \caption{For closed surfaces, our method can generate inside (left) or outside (right) offsets.}
    \label{fig:inside-outside}
\end{figure}

\paragraph{Termination}
The optimization terminates if the average $\sigma_\text{max}$ does not decrease by more than $0.5^\circ$ across iterations, if both the maximum and mean distance error do not decrease by more than $0.5\% \delta$, or after 10 iterations.

\paragraph{Optimization Overview.}
The optimization algorithm proceeds in 3 steps (\Cref{fig:siggraph_logo_adaptive}): we estimate a spatially adaptive distance $\hat{\delta}$ that is as close as possible to $\delta$ (\Cref{sec:distance-adaptation}), we then use a topological marching front algorithm to modify the offset to be as close as possible to $\hat{\delta}$ (\Cref{sec:conservative-estimation}), and finally optimize it with a set of topological and geometrical local operations (\Cref{sec:off-opt}). 

\subsubsection{Distance Adaptation}
\label{sec:distance-adaptation}

\begin{figure}
    \centering\footnotesize
    \includegraphics[width=0.32\linewidth]{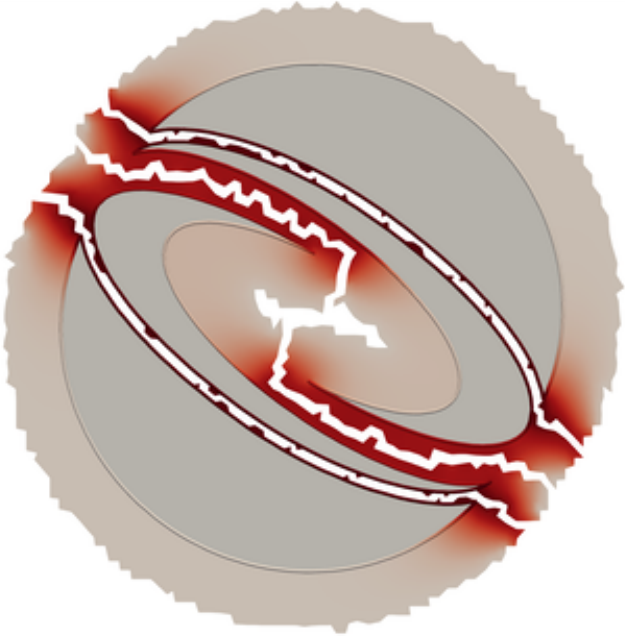}\hfill
    \includegraphics[width=0.32\linewidth]{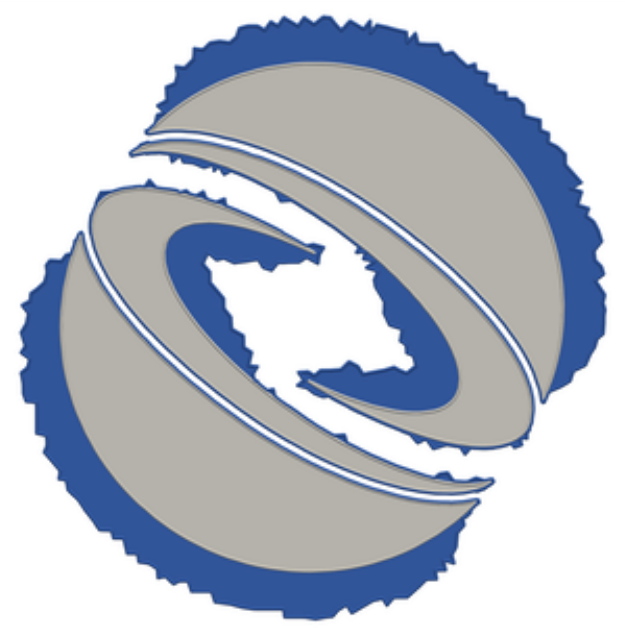}\hfill
    \includegraphics[width=0.32\linewidth]{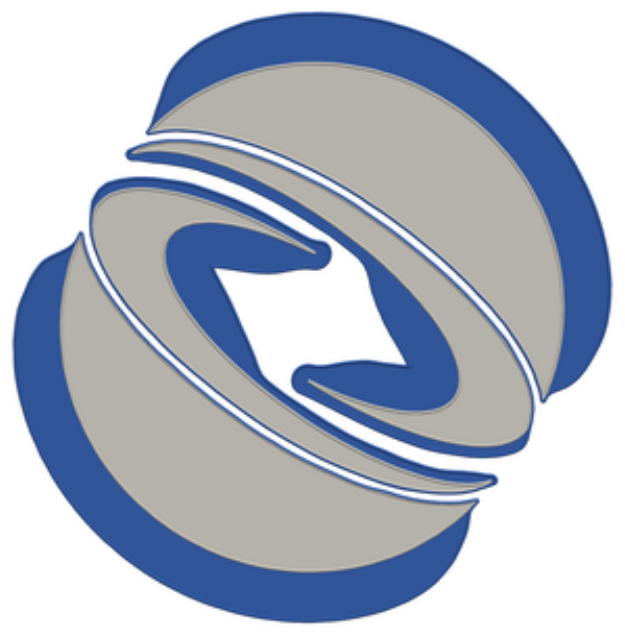}\\
    \parbox{0.32\linewidth}{\centering Distance Adaptation}\hfill
    \parbox{0.32\linewidth}{\centering Conservative Estimation}\hfill
    \parbox{0.32\linewidth}{\centering Optimization}
    \caption{Overview of our topological offset optimization.}
    \label{fig:siggraph_logo_adaptive}
\end{figure}

As we enforce the infinitesimal offset topology, it might happen that two different parts of the offset collide during optimization (\Cref{fig:knot_adaptation}). We do the best effort to avoid such collisions by locally reducing the offset distance.
After initializing the topological offset, we greedily expand the offset volume without changing its topology.
We then compute the distance of the expanded volume's boundary to the input, propagate it back to the input, and eventually use this newly computed distance as our spatially varying distance $\hat{\delta}$.

\paragraph{Greedy Expansion}
To estimate the maximum possible offset distance, we grow the offset volume using a marching front algorithm.
We add tetrahedra to the offset volume (the volume enclosed by the offset surface and the input) if they do not modify the topology of the offset surface. That is, a tetrahedron may only be added to the offset volume if it intersects the surface in one, two, or three faces and their incident edges and vertices. 
We initialize a queue with all tetrahedra that are face adjacent to the offset volume. If a tetrahedron is added to the volume, its face-adjacent neighbors are added to the queue.

A tetrahedron $t$ can only be added to the offset volume if it is within the offset distance $\delta$. 
We conservatively approximate $\delta$ by enclosing $t$ into a sphere (\Cref{fig:distance_estimation}). If the sphere is farther than $\delta$, $t$ is outside. In the other case, we decompose $t$ into 8 sub-spheres and repeat the procedure until either all spheres touching $t$ are outside, one sphere is inside, or the radius is smaller than $10\%\delta$. In the latter cases, we conservatively consider $t$ to be inside. The output of this procedure is a region of space around the input that is homeomorphic to the offset volume.

\begin{figure}
    \centering\footnotesize
    \includegraphics[width=\linewidth]{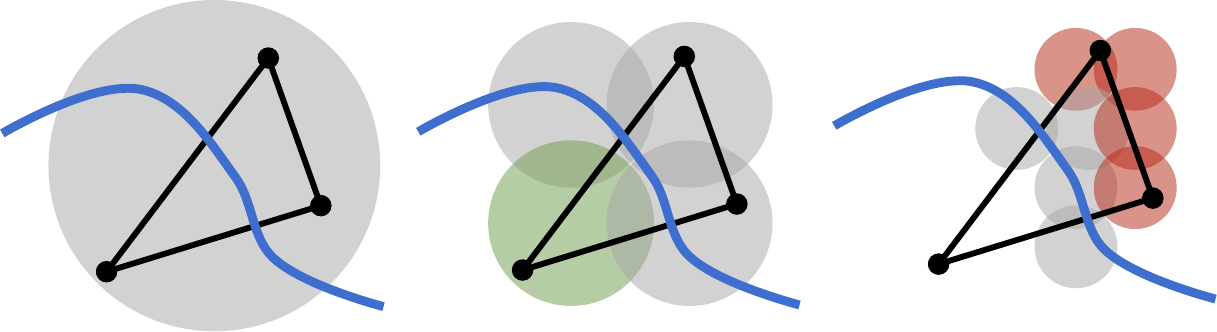}
    \caption{Circle approximation of a triangle to check if it is inside (green) or outside (red) the offset volume. A circle (grey) is subdivided if it intersects the offset surface (blue).}
    \label{fig:distance_estimation}
\end{figure}

\paragraph{Distance Propagation}
We start by assigning the target distance $\delta$ to every vertex on the boundary of the expanded offset volume. Then, we select every tetrahedron vertex adjacent to the expanded offset volume, compute the distance between its barycenter and the input, and update the assigned distance value by selecting the smallest of the values. 
To propagate this distance from the boundary to the input surface, we use harmonic interpolation on the background mesh. 
We store the solution on the vertices of the input simplicial complex as adapted offset distance $\hat{\delta}$.

We discard the expanded offset volume after the distance propagation and re-compute it using the new $\hat{\delta}$ to improve the geometry of the topological offset before optimization.

\begin{figure}
    \centering\footnotesize
   \parbox{.02\linewidth}{\rotatebox{90}{\centering $\delta_\text{avg}$}}\hfill\hfill
    \parbox{.97\linewidth}{\includegraphics[width=\linewidth]{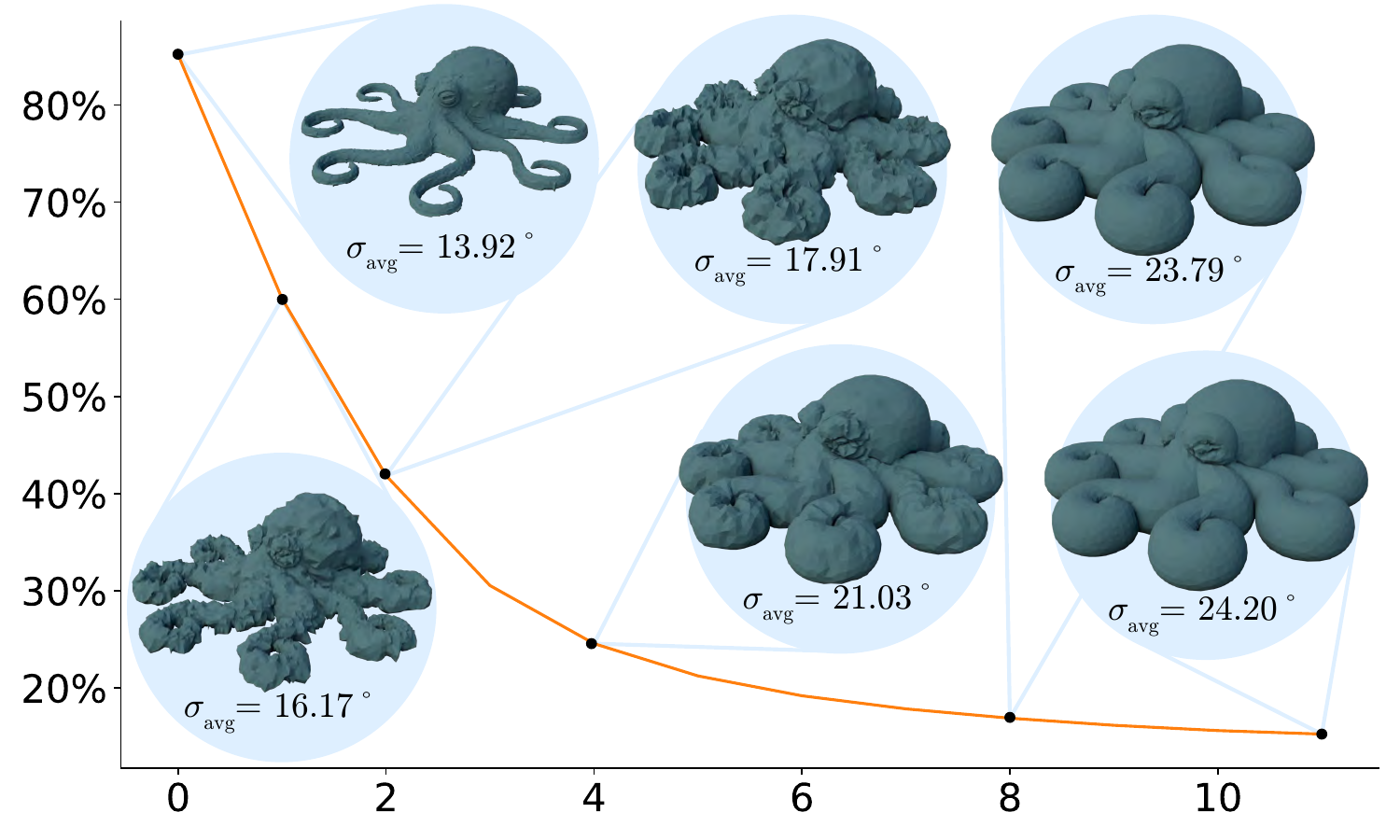}}
    \parbox{.02\linewidth}{~}\hfill\hfill
    \parbox{.97\linewidth}{\centering Iterations}
    \caption{Iterations of the optimization of a topological offset without proper initialization. While the algorithm manages to lower $\delta_\text{avg}$, it plateaus around $15\%$. With proper initialization, it converges to $1.6\%$ in 5 iterations (\Cref{fig:optimization-with-init}).}
    \label{fig:optimization-without-init}
\end{figure}

\subsubsection{Conservative Estimation}
\label{sec:conservative-estimation}
Our goal is to expand the offset volume such that its boundary (the offset surface) is as close as possible to the adapted offset distance $\hat{\delta}$.
We expand the offset volume using the marching front method described in \Cref{sec:distance-adaptation}, but we modify the distance approximation to make it conservative.
A tetrahedron $t$ is only added if it is completely within distance $\hat{\delta}$. Similar to \Cref{sec:distance-adaptation}, we measure the distance by enclosing $t$ into a sphere. If the sphere is within $\hat{\delta}$, $t$ is inside. In the other case, we decompose $t$ into 8 sub-spheres and repeat the procedure until either all spheres touching $t$ are inside, one sphere is outside, or the radius is smaller than $10\%\delta$. In the latter cases, we conservatively consider $t$ to be outside. 
Note that here, the expansion is conservative and not greedy as in \Cref{sec:distance-adaptation}.

The conservative estimation does not just improve performance, but also avoids the optimization to get stuck in local minima (\Cref{fig:optimization-without-init}).

\begin{figure}
    \centering\footnotesize
   \parbox{.02\linewidth}{\rotatebox{90}{\centering $\delta_\text{avg}$}}\hfill\hfill
    \parbox{.97\linewidth}{\includegraphics[width=\linewidth]{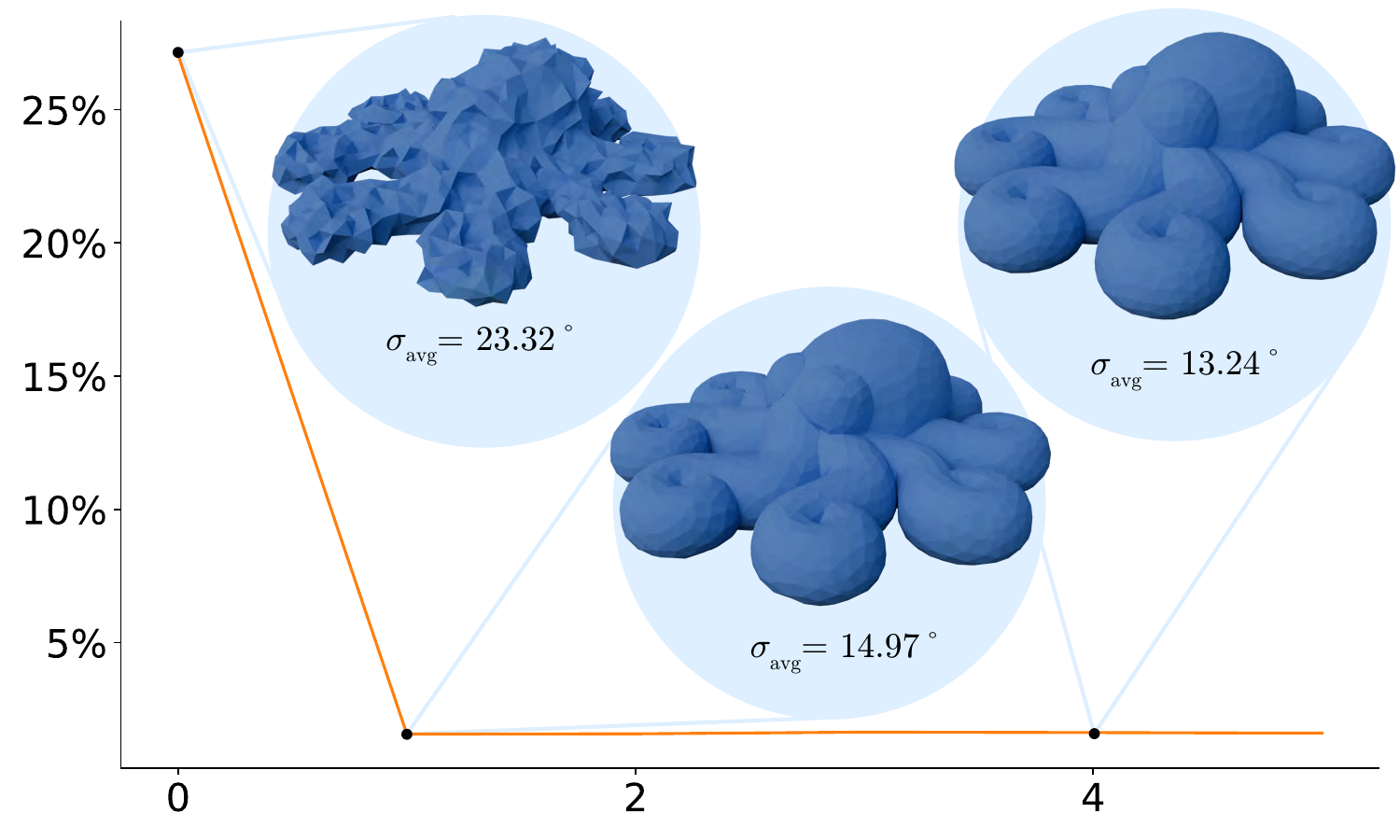}}
    \parbox{.02\linewidth}{~}\hfill\hfill
    \parbox{.97\linewidth}{\centering Iterations}
    \caption{Iterations of the optimization of a topological offset. The average relative distance error $\epsilon_\text{avg}$ and normal deviation $\sigma_\text{avg}$ converge very fast when the offset was properly initialized.}
    \label{fig:optimization-with-init}
\end{figure} 

\begin{figure*}
    \centering\footnotesize
    \rotatebox{90}{\centering dalek}
    \parbox{.18\linewidth}{\includegraphics[width=\linewidth]{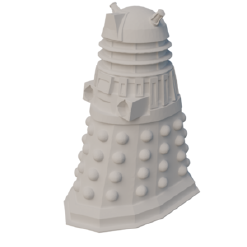}}
    \parbox{.18\linewidth}{\includegraphics[width=\linewidth]{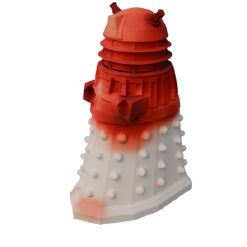}}
    \parbox{.18\linewidth}{\includegraphics[width=\linewidth]{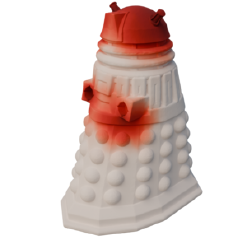}}
    \parbox{.18\linewidth}{\includegraphics[width=\linewidth]{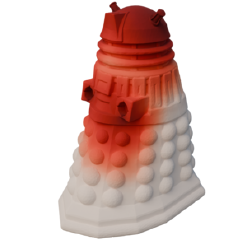}}
    \parbox{.18\linewidth}{\includegraphics[width=\linewidth]{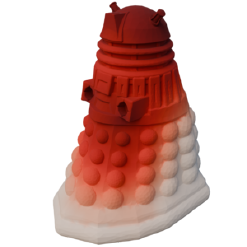}}\\
    \rotatebox{90}{\centering rooster}
    \parbox{.18\linewidth}{\includegraphics[width=\linewidth]{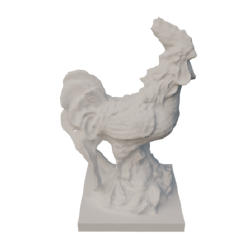}}
    \parbox{.18\linewidth}{\includegraphics[width=\linewidth]{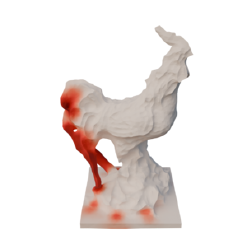}}
    \parbox{.18\linewidth}{\includegraphics[width=\linewidth]{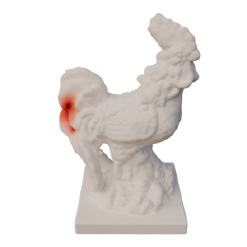}}
    \parbox{.18\linewidth}{\includegraphics[width=\linewidth]{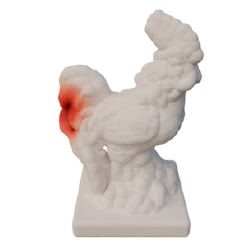}}
    \parbox{.18\linewidth}{\includegraphics[width=\linewidth]{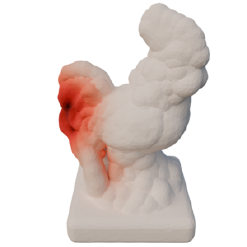}}\\
    \rotatebox{90}{\centering lamp}
    \parbox{.18\linewidth}{\includegraphics[width=\linewidth]{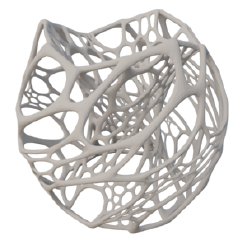}}
    \parbox{.18\linewidth}{\includegraphics[width=\linewidth]{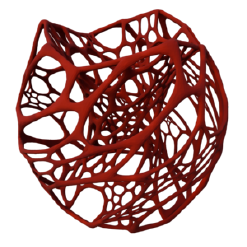}}
    \parbox{.18\linewidth}{\includegraphics[width=\linewidth]{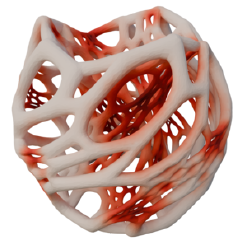}}
    \parbox{.18\linewidth}{\includegraphics[width=\linewidth]{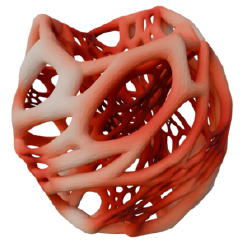}}
    \parbox{.18\linewidth}{\includegraphics[width=\linewidth]{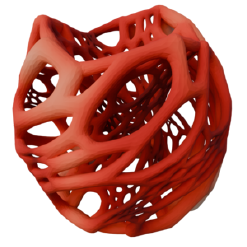}}\\
    \parbox{0.18\linewidth}{\centering input}
    \parbox{0.18\linewidth}{\centering -1\%}
    \parbox{0.18\linewidth}{\centering 1\%}
    \parbox{0.18\linewidth}{\centering 2\%}
    \parbox{0.18\linewidth}{\centering 4\%}
    \caption{Our topological offsets can handle a large variety of offset distances. Regions where the offset distance was adapted are colored in red.}
    \label{fig:topo-offsets-with-varying-distances}
\end{figure*}

\subsubsection{Optimization}
\label{sec:off-opt}

Our optimization algorithm is composed of three steps, which are repeated until convergence: (1) we update a sizing field, (2) we modify the offset mesh, improving its quality and moving its vertices to the desired distance from the input, and (3) we modify the background mesh, increasing its quality. In all iterations, the meshes are modified using a set of local operations, following \cite{botsch2004remeshing}: our algorithm iterates passes of splits, collapse, swaps, and vertex relocation. These operations are always performed on the background mesh: if the algorithm tries to split an edge of the offset mesh, this operation is applied to the corresponding edge of the background mesh or vice versa.
Note that average normal deviation and relative distance error converge very fast due to the conservative estimation (\Cref{fig:optimization-with-init}).

\paragraph{Invariants} 
The operations are executed only if their effect on the mesh satisfies the following invariants: (I1) the input surface and the boundary are not modified, (I2) the orientation of all tetrahedra in the background mesh is preserved (tested using the exact predicate in \cite{Shewchuck:1996}) and (I3) the topology of the offset and input surface is preserved \cite{Vivodtzev2010}. 

\begin{theorem}
Let $\M$ be a background mesh containing a simplicial embedding $S$ and an offset surface $O$. The mesh $\M'$ computed after performing any sequence of local operations satisfying the invariants I1, I2, and I3 contains a new surface $O'$ homeomorphic to $O$ that does not intersect $S$. 
\label{thm:op}
\end{theorem}

\begin{proof}
$O'$ is homeomorphic to $O$ due to the explicit avoidance of operations changing its topology (I3), we refer to \cite{Vivodtzev2010} for details. $O'$ cannot intersect $S$ because I1 and I2 imply that there is always a continuous bijection between the points in $\M$ and $\M'$ \cite{Lipman2014}.
\end{proof}

It follows from \Cref{thm:op} that topological offsets are preserved by our mesh optimization.

\paragraph{Step 1: Sizing Field Update}
Edge split and collapse operations are driven by a sizing field that is defined on each edge of the offset mesh and is initialized with the current length of each edge. 

In each update pass, if one of the incident triangles has a shape regularity below $0.5$ or a normal deviation above the user-defined maximum $\sigma_\text{max}$, we divide the target length by two. 
If shape regularity is below $0.5$ and the normal deviation is above $\sigma_\text{min}$, we increase the target length by $1.5$. After this update, to keep the sizing field smooth, we cap the target length to $1.5$ times the length of any adjacent edge.

\paragraph{Step 2: Local operations}
Every edge that has a length greater than $4/3$ of its target length is split. The new vertex is positioned at the center of the edge. All edges that are shorter than $3/4$ of their target length are collapsed. We perform half-edge collapses. While there are no extra conditions for splits, a collapse is only performed if the user-defined maximum normal deviation is not exceeded. Both operations are scheduled according to the current edge length but for splits, long edges are prioritized, while for collapse, the short ones are considered first.
Edges are swapped if the operation increases the minimal triangle shape regularity (\Cref{subsec:quality_metrics}) of the two adjacent triangles. The swap is not performed if the normal deviation before the operation is below the user-defined maximum and would be above afterward. Long edges are prioritized in the swap operation.
We adapt the vertex relocation method proposed in \cite{zint2023feature} to work with a spatially varying offset distance. We compute the offset distance for a vertex as the area-weighted average of offset distances at the sample points of the adjacent triangles.
If the computed position would cause a tetrahedron from the background mesh to be flipped, we perform a binary search along the way to the computed position to find a valid position. If no valid position can be found, we do not move the vertex.

\Cref{fig:topo-offsets-with-varying-distances} shows different results of our topological offsets for different distances on 3 models from Thingi10k \cite{Thingi10K} where the distance was locally adapted (\Cref{tab:topo-offsets-examples-table} shows the statistics). 
All the depicted examples have an average relative distance error below $3\%$, an average normal deviation below $20^\circ$, and an average triangle shape regularity of at least $0.78$. The quality metrics are explained in \Cref{subsec:quality_metrics}.

\paragraph{Step 3: Embedding Optimization}
We use the optimization scheme that was presented in \cite{hu2018tetrahedral} with two minor modifications (in addition to the aforementioned invariants) to make it more efficient, as we are not interested in obtaining a background mesh of very high quality: we only want the background mesh to not hinder the movement of the offset. First, we trigger the update of the sizing field if the tetrahedron AMIPS energy is above 100 (instead of 8) and limit the target edge length to three times the length of any adjacent edge. Second, we only optimize the two-ring neighborhood of a tetrahedron with AMIPS above 100. 

\begin{figure}
    \centering\footnotesize
    \includegraphics[width=0.49\linewidth]{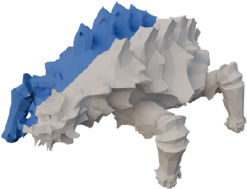}\hfill
    \includegraphics[width=0.49\linewidth]{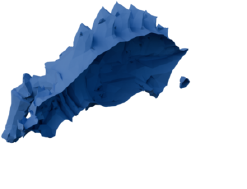}
    \parbox{0.49\linewidth}{\centering input}\hfill
    \parbox{0.49\linewidth}{\centering offset volume}
    \caption{Our method fails if an edge split during the offset initialization causes tetrahedra to invert due to numerical inaccuracy. In this example, we try to compute two offsets, one at $10^{-11}$ and a second one at $10^{-12}$.}
    \label{fig:failure-case}
\end{figure}

\subsection{Robustness and Failure Cases}
\label{sec:robustness}

We now analyze our algorithm from a robustness perspective, discussing precisely which guarantees it provides and what are the potential failures.

\paragraph{Step 1 and 2.} The decision of where to split is purely topological. However, the background mesh might flip after the split if an edge is shorter than the rounding error (Figure \ref{fig:failure-case}). This is not a practical concern, but a fully robust solution could be obtained by using a hybrid floating point/rational representation, following the same idea in \cite{hu2018tetrahedral}.

\paragraph{Step 3.} The invariants in the optimization are either purely topological (I1 and I3), or checked using exact predicates \cite{Shewchuck:1996} (I2). While there are no guarantees that the prescribed distance will be obtained, this step cannot fail and will always terminate after 10 iterations.

\paragraph{Properties}
The generated offsets are thus guaranteed to be homeomorphic to an infinitesimal offset, do not intersect the input simplicial complex, and their embedding is free of self-intersections. We show a large-scale validation of our implementation in \Cref{sec:results}. We note that using a tetrahedral background mesh makes these strong guarantees possible, enabling us to reduce challenging checks (self-intersections, topological correctness) to an exact \texttt{Orient3D} predicate. 

\section{Results}
\label{sec:results}

\begin{figure}
    \centering\footnotesize
    \parbox{.02\linewidth}{\rotatebox{90}{\centering Number of models}}\hfill
    \parbox{.968\linewidth}{\includegraphics[width=\linewidth,trim={0 1.45cm 0 0},clip]{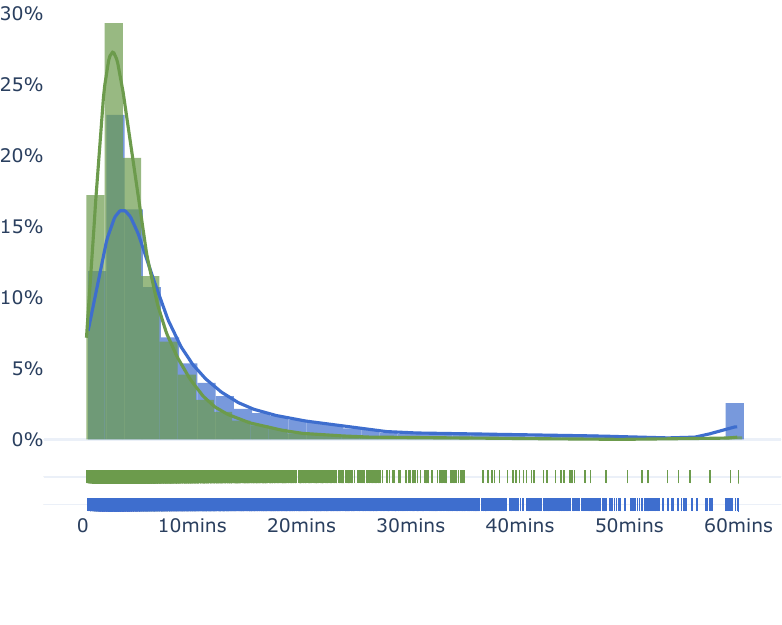}}
    \parbox{\linewidth}{\centering Runtime}\hfill
    \caption{Runtime of topological (blue) and finite (green) offsets.}
    \label{fig:thingi10k-runtime}
\end{figure}

We created topological offsets for the entire Thingi10k dataset~\cite{Thingi10K}, which was embedded in a background mesh using TetWild~\cite{hu2018tetrahedral} with default settings. We run our experiments on cluster nodes with a Xeon E5-2690 v2 @ 3.00GHz. We use the winding number to identify the outer part of the surface and create a one-sided offset. We skipped the 233 meshes where the winding number failed to identify a closed internal volume. This left $9\,767$ meshes in our experiments. We cap the runtime at 24 hours or stop when the convergence criteria from \Cref{sec:off-opt} (\textit{Termination}) are met. We compute offsets on the dataset with a target distance $\delta = 4\%$ relative to the bounding box size. We are not aware of any existing method that produces topological offsets, and therefore, we cannot directly compare with any previous work. To enable direct comparisons, we introduce a minor variant of our algorithm to create finite offsets: this variant is described and compared with state-of-the-art offset methods in Section (\Cref{subsec:finite-offsets}). 

\paragraph{Large-Scale Testing.} 
Our algorithm produces a valid topological offset, embedded in a valid background mesh, for all models of the dataset except for one, where no edge split could be performed without generating inverted elements. For more details on that matter, see \Cref{sec:robustness}. $55.24\%$ of the models finish within 6 minutes, and only 275 models (less than 3\%) take more than one hour (\Cref{fig:thingi10k-runtime}); these are highly complex models (\Cref{fig:slow}). 
The overall memory consumption is low, considering that our method processes tetrahedral meshes. $9\,764$ models use less than 16 GB of memory, and the remaining 3 use less than 64 GB.

\begin{figure}
    \centering
    \includegraphics[width=\linewidth]{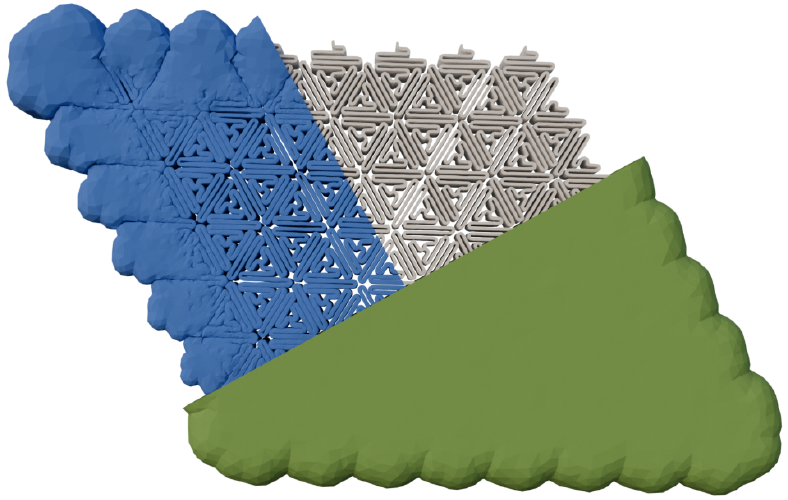}
    \caption{For models with high geometric fidelity, the topological offset (blue) contains way more triangles, whereas finite offsets (green) remove geometric details.}
    \label{fig:slow}
\end{figure}

\begin{figure}
    \centering\footnotesize
    \parbox{.02\linewidth}{\rotatebox{90}{\centering Number of models}}\hfill
    \parbox{.968\linewidth}{\includegraphics[width=\linewidth,trim={0 0.5cm 0 0},clip]{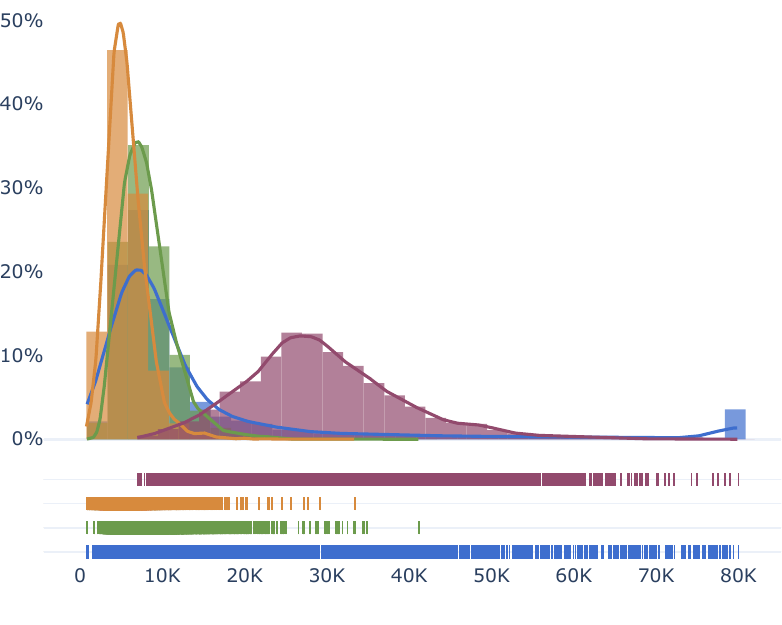}}
    \parbox{\linewidth}{\centering Number of triangles}\hfill
    \caption{Number of offset triangles of Alpha Wrapping (red), FPO (yellow), our finite offsets (green), and our topological offsets (blue).}
    \label{fig:thingi10k-triangle-number}
\end{figure}

\begin{figure}
    \centering\footnotesize
    \parbox{.02\linewidth}{\rotatebox{90}{\centering Number of models}}\hfill
    \parbox{.968\linewidth}{\includegraphics[width=\linewidth,trim={0 0.5cm 0 0},clip]{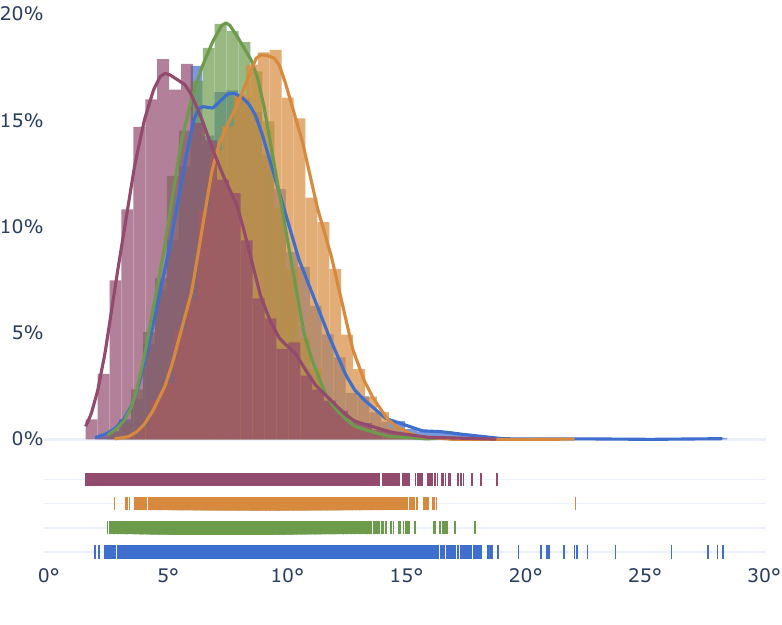}}
    \parbox{\linewidth}{\centering Average normal deviation}\hfill
    \caption{Average normal deviation of Alpha Wrapping (red), FPO (yellow), our finite offsets (green), and our topological offsets (blue).}
    \label{fig:thingi10k-normal-deviation}
\end{figure}

\begin{figure}
    \centering
    \includegraphics[width=0.49\linewidth]{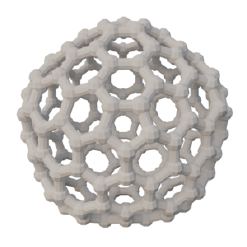}\hfill
    \includegraphics[width=0.49\linewidth]{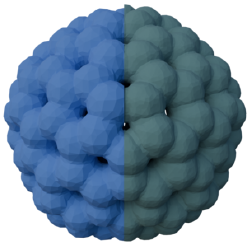}
    \caption{The minimal edge length can be adapted (teal) if the normal deviation is large (blue) for a given offset distance and model (white).}
    \label{fig:wire-sphere}
\end{figure}

\paragraph{Density and Quality} Our topological offsets can have many triangles (\Cref{fig:thingi10k-triangle-number}) as they always have the topology of an infinitesimal offset and, therefore, maintain the details present in the input (\Cref{fig:slow}). There are few results (\Cref{fig:thingi10k-normal-deviation}) with an average normal deviation larger than $20^\circ$ as the normal deviation error depends both on the edge length and the offset distance. Our experiments show that these are models where we adapt the offset distance everywhere. By reducing the edge length, we can achieve the target normal deviation even for the adapted offset distance (\Cref{fig:wire-sphere}).

\paragraph{Distance Adaptation}
The distance adaptation relies on a heuristic that might fail under certain circumstances, since the greedy expansion (\Cref{sec:optimization}) might overestimate the offset. The geometry of the estimation and the actual topological offset might be different, and in such a case, the distance might not be adapted properly. In \Cref{fig:clip}, we adapted the distance for the $2\%$ offset, while no adaptation was necessary for $7\%$. The greedy expansion of $4\%$ overestimates the geometry and is more similar to the one of the $7\%$ distance, where no adaptation is necessary; this leads to very close offsets but never intersecting, as we prohibit that with exact predicates (\Cref{sec:optimization}). Note that, for this example, the input and the offsets are of genus 0, so for large distances, the offset converges to a sphere.

\begin{figure}
    \centering\footnotesize
    \includegraphics[width=0.24\linewidth]{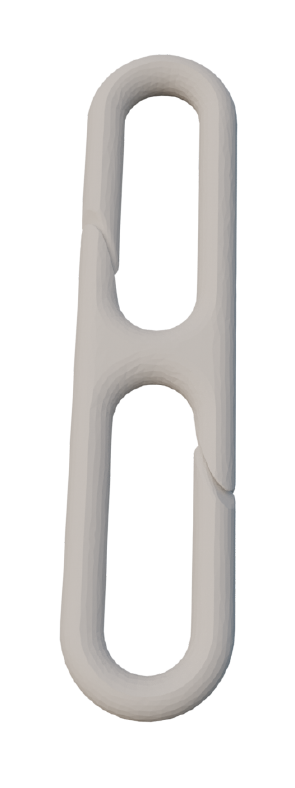}\hfill
    \includegraphics[width=0.24\linewidth]{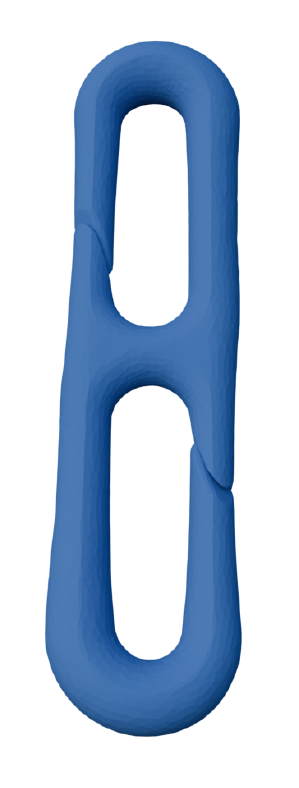}\hfill
    \includegraphics[width=0.24\linewidth]{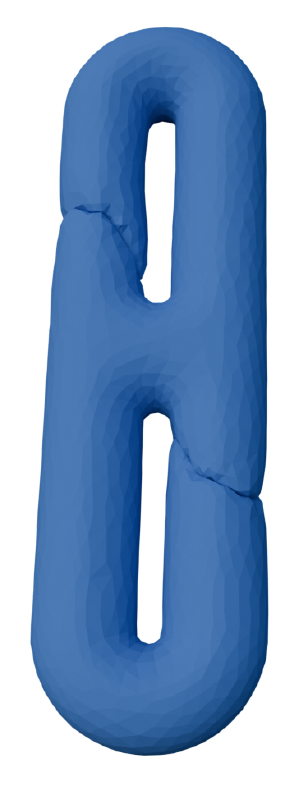}\hfill
    \includegraphics[width=0.24\linewidth]{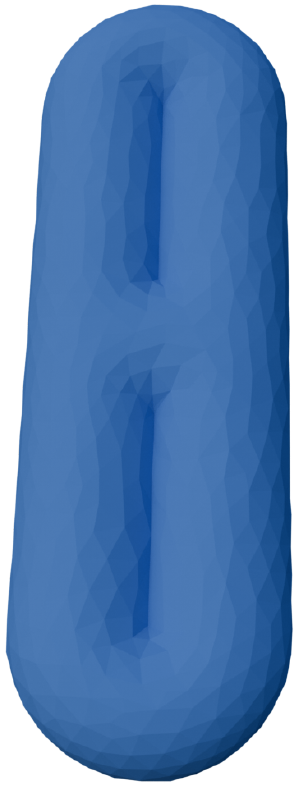}
    \parbox{0.24\linewidth}{\centering input}\hfill
    \parbox{0.24\linewidth}{\centering 2\%}\hfill
    \parbox{0.24\linewidth}{\centering 4\%}\hfill
    \parbox{0.24\linewidth}{\centering 7\%}
    \caption{Topological offsets might be very close (but never intersecting) as the local offset distance adaptation might be inaccurate.}
    \label{fig:clip}
\end{figure}

\section{Applications}

We present several applications of our topological offset: an algorithmic variant for producing finite offsets (\Cref{subsec:finite-offsets}), construction of multiple offset layers which guarantees that the offsets are intersection-free (\Cref{subsec:layered-offsets}), 
and the use of topological offsets to remove non-manifold simplices (\Cref{subsec:manifold-extraction}).

\begin{figure}
    \centering
    \includegraphics[width=\linewidth]{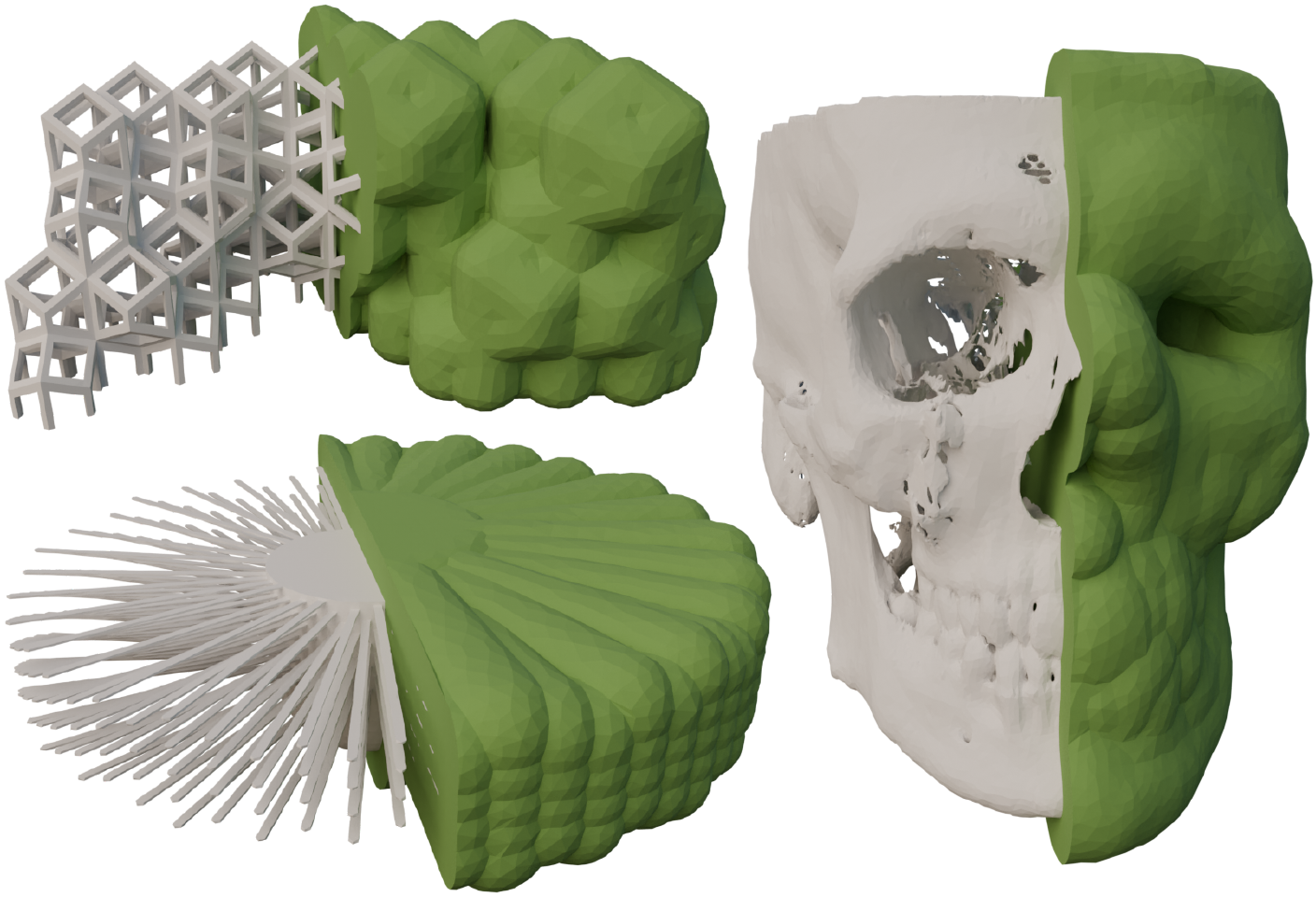}
    \caption{With a small modification, our offsets become finite (like those produced by other methods). Our finite offsets (like our topological ones) are guaranteed to enclose the input and are manifold and self-intersection-free.}
    \label{fig:finite_examples}
\end{figure}

\subsection{Finite Offsets}
\label{subsec:finite-offsets}
We can adapt our algorithm to produce finite offsets with minor modifications (\Cref{fig:finite_examples}). Note that in this case our algorithm does not produce an offset homeomorphic to an infinitesimal one and, similarly to all other finite offset methods, is not guaranteed to compute an offset with the same topology as the exact finite offset with the prescribed distance, it is only an approximation. However, our finite offsets are still manifold, self-intersection-free, and enclose the input.

For finite offsets, we expand the topological offset (\Cref{sec:conservative-estimation}) without any topological condition. After we grow the region to the desired distance, we generate another topological offset around it so that we can ensure that the offset is manifold, even if the expanded offset volume is not. While the topology of the topological offsets is independent of the resolution of the embedding, this is not true for finite offsets.
We compute finite offsets for the entire Thingi10k dataset and observe that the running time is reduced compared to topological offsets due to the reduced geometric complexity (\Cref{fig:thingi10k-runtime}); $63.7\%$ of the models finish within 5 minutes. Just 21 models ($0.2\%$) take more than one hour. None of our finite offsets has a normal deviation larger than $20^\circ$ (\Cref{fig:thingi10k-normal-deviation}). 
 
\subsubsection{Comparison}
\label{subsec:comparison}

We compare our finite offsets (\Cref{subsec:finite-offsets}) with Feature-Preserving Offsets (FPO)~\cite{zint2023feature}, and 3D Alpha Wrapping in CGAL~\cite{cgal:achpr-aw3-24a}. We run FPO with default parameters except for the normal deviation which we set to $15^\circ$. For Alpha Wrapping, we set $\alpha = \delta / 5$, which he chose to produce similar results.
Note that Alpha Wrapping is not an offsetting method and, therefore, does not claim to be feature-preserving or topologically correct. We chose this method to compare against because it comes with similar guarantees to ours (watertight, orientable, and strictly contains the input).

\paragraph{Normal Deviation}
All four methods produce meshes with similar average normal deviation (\Cref{fig:thingi10k-normal-deviation}). Alpha Wrapping performs slightly better, which can be explained by the larger number of triangles (\Cref{fig:thingi10k-triangle-number}).

\begin{figure}
    \centering\footnotesize
    \parbox{.02\linewidth}{\rotatebox{90}{\centering Number of models}}\hfill
    \parbox{.968\linewidth}{\includegraphics[width=\linewidth,trim={0 0.5cm 0 0},clip]{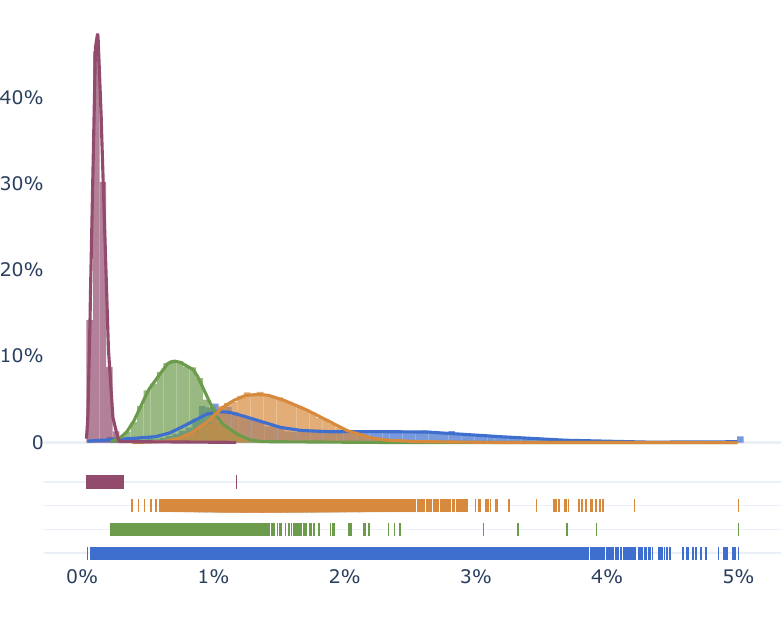}}
    \parbox{\linewidth}{\centering Relative average distance error}\hfill
    \caption{Relative average distance error of Alpha Wrapping (red), FPO (yellow), our finite offsets (green), and our topological offsets (blue).}
    \label{fig:thingi10k-distance}
\end{figure}

\paragraph{Offset Distance}
We report the average distance error relative to the user-defined target distance in \Cref{fig:thingi10k-distance}. Again, Alpha Wrapping (red) performs best, but it also has the largest amount of triangles, which influences the metric. Our method is adaptive to curvature and therefore places more triangles where normal deviation is high, while Alpha Wrapping has a uniform triangle distribution and therefore more triangles in flat regions with zero normal deviation.

\begin{figure}
    \centering\footnotesize
    \includegraphics[width=0.99\linewidth]{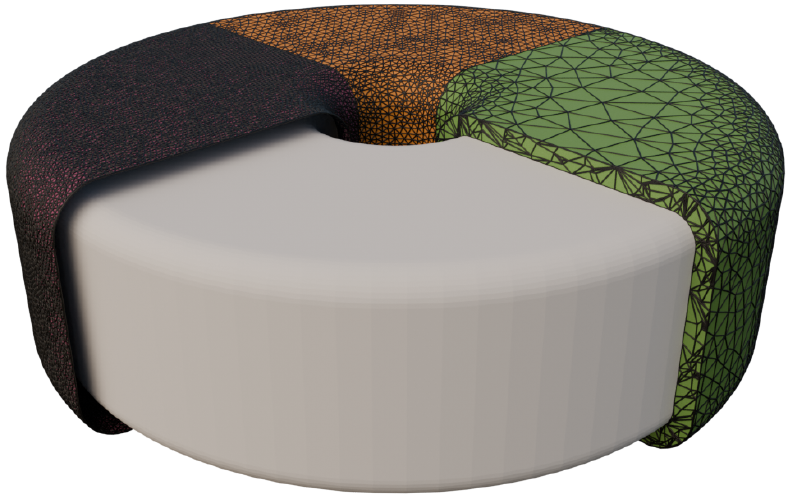}
    \parbox{.32\linewidth}{\centering \cite{cgal:achpr-aw3-24a}}\hfill
    \parbox{.32\linewidth}{\centering \cite{zint2023feature}}\hfill
    \parbox{.32\linewidth}{\centering ours}
    \caption{All three methods perform equally well on simple models. \cite{cgal:achpr-aw3-24a} in red, \cite{zint2023feature} in orange, and ours in green.}
    \label{fig:comparison_simple_model}
\end{figure}

\begin{figure}
    \centering\footnotesize 
    \includegraphics[width=0.32\linewidth]{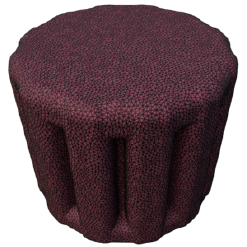}\hfill
    \includegraphics[width=0.32\linewidth]{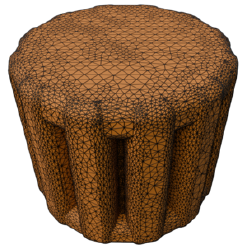}\hfill
    \includegraphics[width=0.32\linewidth]{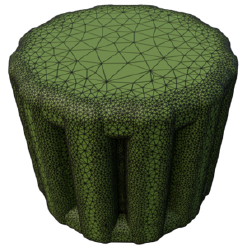}
    \parbox{.32\linewidth}{\centering\cite{cgal:achpr-aw3-24a}}\hfill
    \parbox{.32\linewidth}{\centering \cite{zint2023feature}}\hfill
    \parbox{.32\linewidth}{\centering ours}
    \caption{Our method is the only one with a varying sizing field.}
    \label{fig:comparison_adaptivity}
\end{figure}

\begin{figure}
    \centering\footnotesize
    \includegraphics[width=\linewidth]{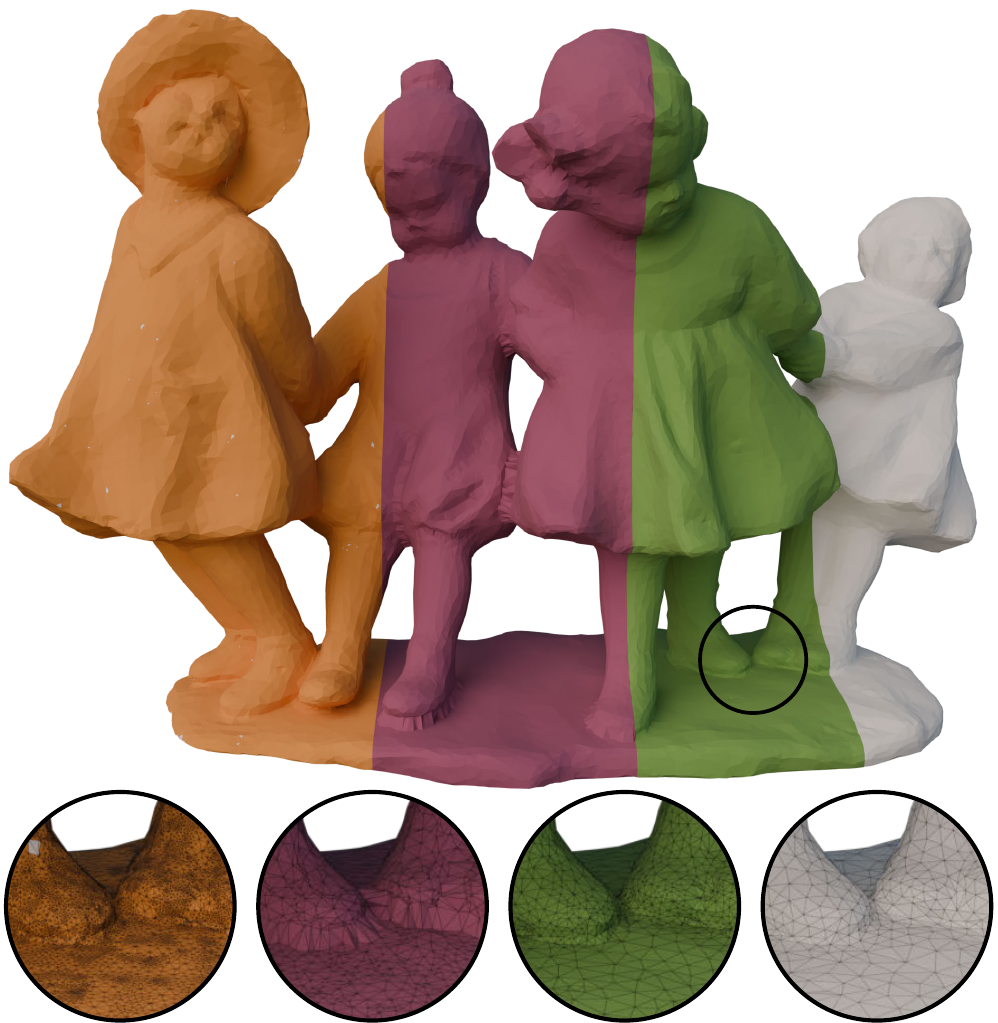}
    \parbox{.24\linewidth}{\centering \cite{zint2023feature}}\hfill
    \parbox{.24\linewidth}{\centering \cite{cgal:achpr-aw3-24a}}\hfill
    \parbox{.24\linewidth}{\centering ours}\hfill
    \parbox{.24\linewidth}{\centering input}
    \caption{Our method (green) is well suited for small offset distances. We use $\delta = 0.01\% $ and set $\ell_\text{min}$ to the average input edge length (white) for ours and FPO (orange). For Alpha Wrapping (red), we set $\alpha = 50 \delta$ to achieve comparable edge lengths.}
    \label{fig:small-distance}
\end{figure}

\begin{figure}
    \centering
    \includegraphics[width=0.32\linewidth]{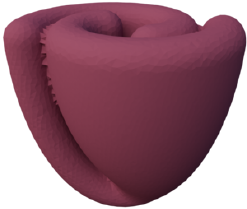}\hfill
    \includegraphics[width=0.32\linewidth]{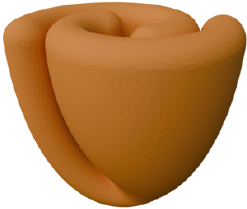}\hfill
    \includegraphics[width=0.32\linewidth]{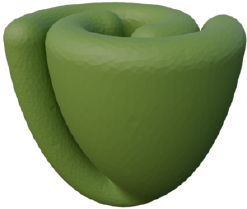}

    \includegraphics[width=0.32\linewidth]{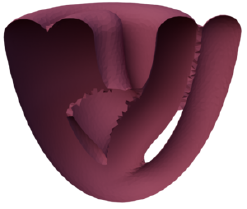}\hfill
    \includegraphics[width=0.32\linewidth]{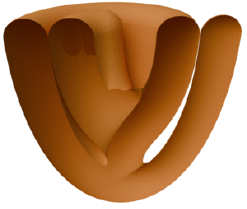}\hfill
    \includegraphics[width=0.32\linewidth]{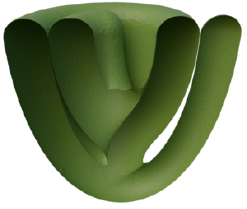}
    \parbox{.32\linewidth}{\centering\footnotesize \cite{cgal:achpr-aw3-24a}}\hfill
    \parbox{.32\linewidth}{\centering\footnotesize \cite{zint2023feature}}\hfill
    \parbox{.32\linewidth}{\centering\footnotesize ours}
    \caption{Offsets generated with different methods.}
    \label{fig:comparison_touching_offsets}
\end{figure}

\paragraph{Qualitative Comparison}
On simple models, like the one in \Cref{fig:comparison_simple_model}, Alpha Wrapping, FPO, and our finite offsets perform equally well. However, our method requires significantly more time to produce similar results. Alpha Wrapping and FPO finish in 2.6 and 33.7 seconds, respectively, while our method needs 242 seconds. This overhead is caused by the background mesh that needs to be updated. Our method returns a tetrahedral mesh in which the input and offset are embedded, while the others only generate surfaces.

Our method uses a sizing field to be adaptive to the offset curvature. Flat regions are not unnecessarily refined (\Cref{fig:comparison_adaptivity}) while the mesh is denser in regions with high curvature compared to the meshes from the other methods.

Our method is well suited for small offset distances (\Cref{fig:small-distance}). In this example, we set   $\ell_\text{min}$ to the input average edge length to avoid unnecessary refinement. Alpha Wrapping introduces artifacts in convex regions (we picked $\alpha$ so that it creates a similar number of triangles as the input). For such a small distance, FPO produces visible intersections (white speckles in \Cref{fig:small-distance}) and overrefines the offset even with the minimal edge length restricted. 

An extreme challenge for offset methods is when two offsets are almost colliding (\Cref{fig:comparison_touching_offsets}).
Alpha Wrapping cannot enter the thin area in between the two offsets. FPO produces a good-looking result from the outside but generates self-intersections on the interior. Our method produces the desired outcome and is free of self-intersections.

\begin{figure}
    \centering
    \includegraphics[width=0.49\linewidth]{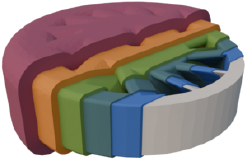}
    \includegraphics[width=0.49\linewidth]{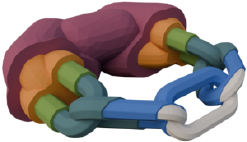}
    \caption{Layered finite offsets at 8\%, 4\%, 2\%, 1\%, and 0.5\% of the bounding box size.}
    \label{fig:gear_layered_offset}
\end{figure}

\subsection{Layered Offsets}
\label{subsec:layered-offsets}
We can construct multiple (finite or topological) non-intersecting offset layers by adding a simple condition to the offset initialization (\Cref{fig:gear_layered_offset}).
The outermost layer must be generated first. For the next layer, the offset volume expansion (\Cref{sec:conservative-estimation}) must not touch any previously generated layer. This guarantees that all layers are intersection-free.

\subsection{Manifold Extraction}
\label{subsec:manifold-extraction}
A common way to remove non-manifold vertices is to duplicate them and displace them in the opposite normal direction. While this method is simple, it comes with several drawbacks. First, it is not guaranteed that a valid normal direction always exists, and therefore it can fail for certain scenarios. Second, if the mesh was embedded in a background mesh, the region around the duplicated vertex needs to be remeshed. For an in-depth discussion of related work, we refer the reader to \cite{attene2009on}.

\begin{figure}
    \centering\footnotesize
    \parbox{.02\linewidth}{\rotatebox{90}{\centering Number of models}}\hfill
    \parbox{.968\linewidth}{\includegraphics[width=\linewidth]{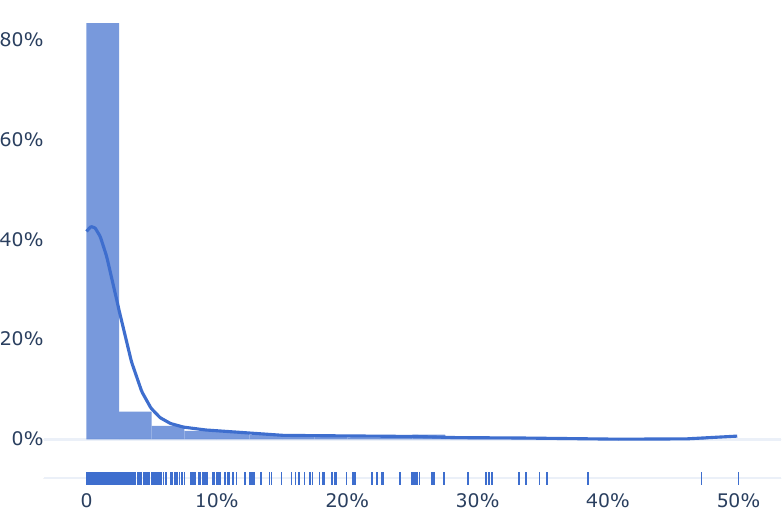}}
    \parbox{\linewidth}{\centering Relative increase of the number of triangles}\hfill
    \caption{Increase in the number of triangles required to make the meshes manifold.}
    \label{fig:manifold-stats}
\end{figure}

We can remove non-manifold simplices from a mesh by constructing a topological offset around them.  First, we detect all non-manifold simplices and consider them as the input for our topological offset.
Second, we remove all simplices from the input mesh within the offset region, obtaining a manifold mesh.

Our algorithm is guaranteed to generate manifold meshes and to keep the embedding valid; however, it might add unnecessary vertices to the surface. To mitigate this effect, we perform a ``clean-up'' after inserting the offset: we collapse edges within the region of the same non-manifold vertex. We only perform collapses that maintain a valid embedding.
Finally, we push all remaining vertices to a user-defined offset distance. If the desired distance would cause tetrahedral inversions, we use a binary search to find a valid position. If we cannot find a valid position, we keep the vertex where it is.

This method has similarities to one of the methods presented in \cite{attene2009on}: this method requires specifying a radius for the removal of the non-manifold points, and different choices of this radius will create different topologies. In contrast, our method is parameter-free and can be implemented robustly using floating-point computation only.

We run the non-manifold removal algorithm on the 1053 models from Thingi10k \cite{Thingi10K}, processed by TetWild \cite{hu2018tetrahedral}, whose surfaces are non-manifold. The algorithm succeeds for all models, i.e., there are no issues with floating point accuracy, no models contain inverted elements, and all surfaces do not intersect and are manifold. In \Cref{fig:manifold-stats}, we summarize the results: for 970 out of the 1053 models~(92\%), our method introduces less than 10\% more triangles on the surface. There are only 4 models where the element count doubles since almost all surface vertices and edges are non-manifold. \Cref{fig:manifold} shows how non-manifold vertices and edges are successfully removed. 

\begin{figure}
    \centering\footnotesize
    \includegraphics[width=0.49\linewidth]{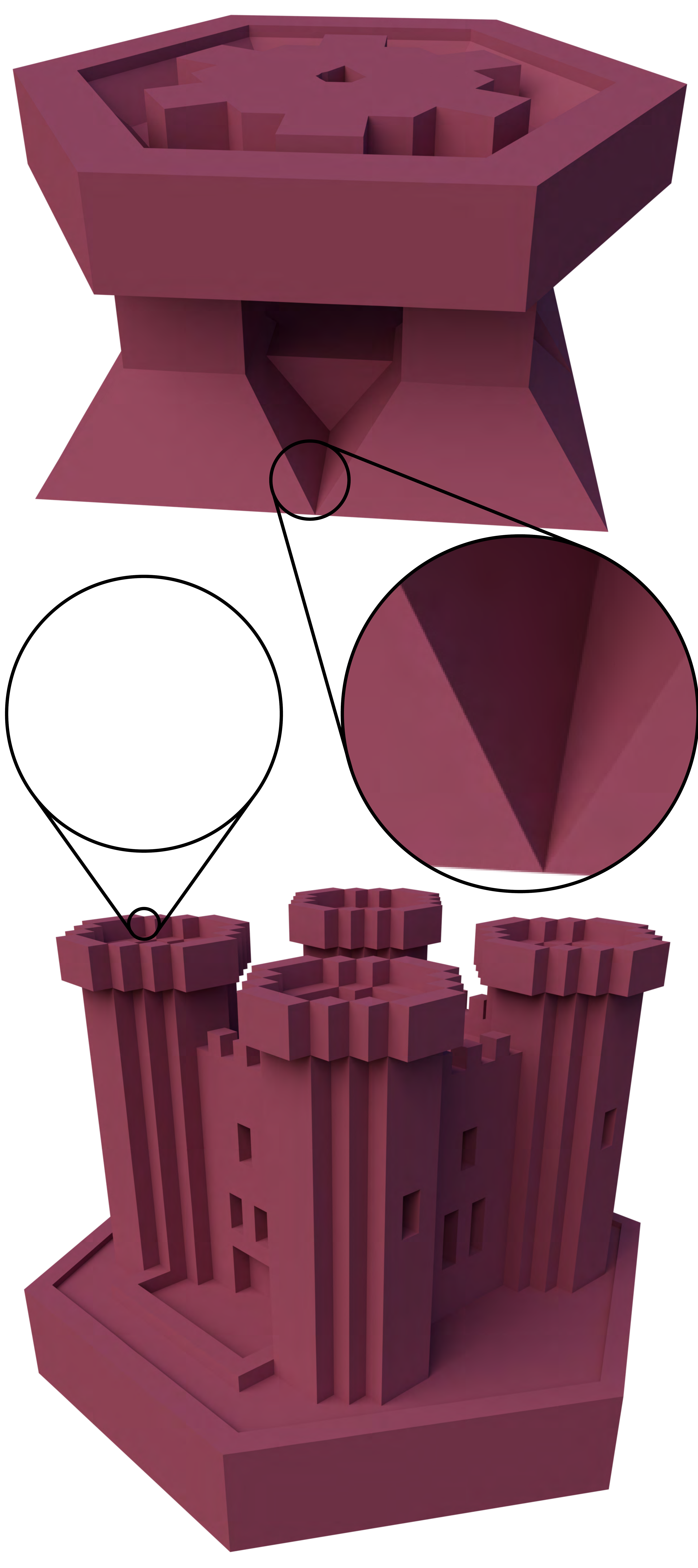}\hfill
    \includegraphics[width=0.49\linewidth]{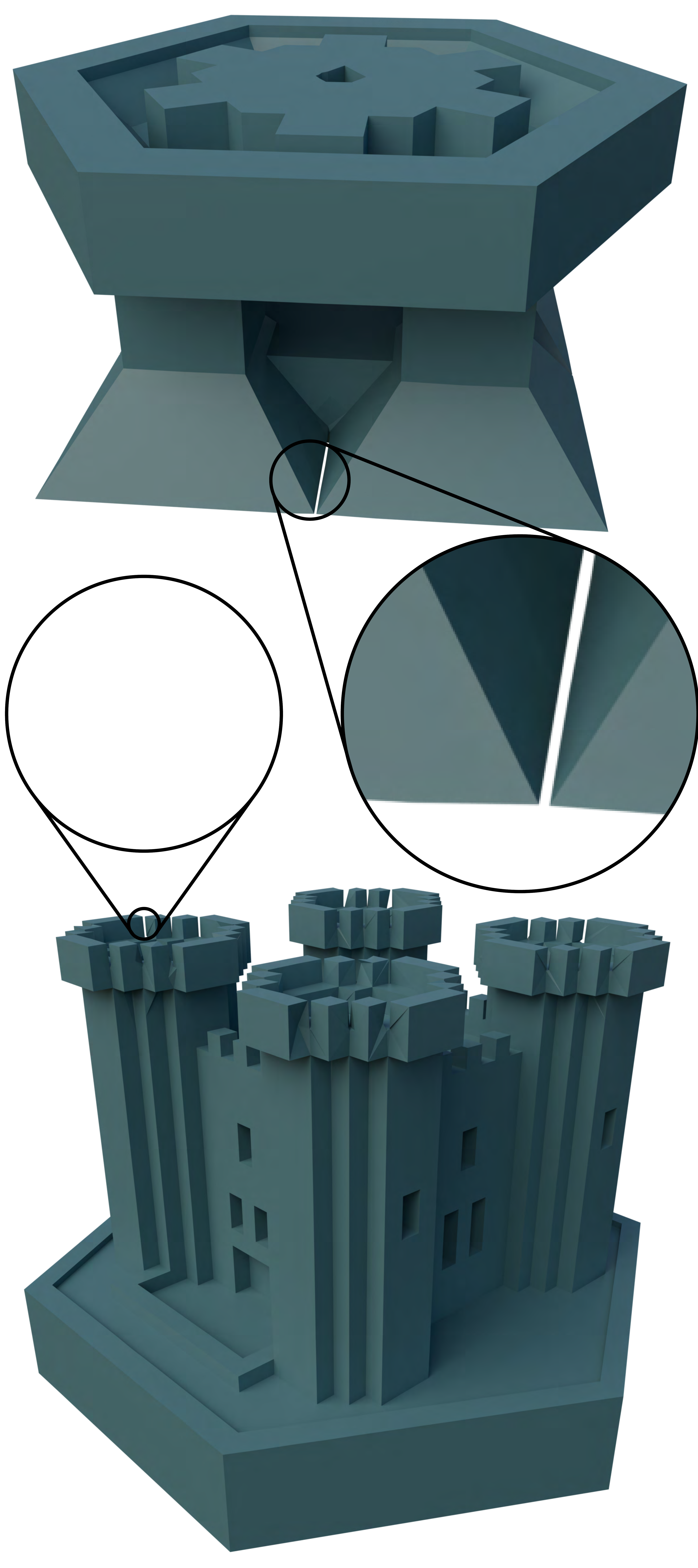}
    \parbox{0.49\linewidth}{\centering non-manifold input}\hfill
    \parbox{0.49\linewidth}{\centering manifold output}
    \caption{Example of removal of non-manifold regions using our algorithm.}
    \label{fig:manifold}
\end{figure}

\section{Concluding Remarks}
We introduced topological offsets, an algorithm for computing them robustly, and demonstrated their relevance in a wide range of graphics applications.

While computationally more expensive than competing finite offset methods, our algorithm generates the unique topology of an infinitesimal offset, 
and guarantees to produce self-intersection-free offsets that strictly enclose the input. 
We use these guarantees to extend our construction to multiple offsets (topological or finite) that inherit the same guarantees and therefore are strictly enclosing each other.

Exploring the use of this approach to create layered offsets (boundary layers) with exponentially increasing thickness and their use for fluid simulation is an exciting avenue for future work. Additionally, we plan to explore parallel or distributed mesh optimization methods to reduce the running time difference compared to other offset methods.

\begin{acks}
This work was supported in part through the NYU IT High Performance Computing resources, services, and staff expertise. This work was also partially supported by the NSF grants OAC-2411349 and IIS-2313156, the NSERC grants DGECR-2021-00461 and RG-PIN 2021-03707, as well as a gift from Adobe Research.
\end{acks}

\begin{table}
    \centering
    \caption{Statistics for the examples presented in \Cref{fig:topo-offsets-with-varying-distances}. The columns display from left to right: mesh name, target offset distance, runtime in seconds, number of embedding tetrahedra, number of offset triangles, average triangle shape regularity, average normal deviation, average relative distance error.}
    \begin{tabular}{|r||r|r|r|r|r|r|}
        \multicolumn{7}{l}{dalek (49874)} \\\hline
        $\delta$ & time (s) & $\#T$ & $\#t$ & $q_{sr,avg}$ & $\sigma_{avg}$ & $\epsilon_{avg}$ \\\hline
        -1\% & $1\,656$ & $   921\,789$ & $102\,086$ & 0.84 &  12 & 1.8\% \\
         1\% & $1\,524$ & $   960\,100$ & $109\,280$ & 0.87 &  11 & 1.2\% \\
         2\% & $1\,140$ & $   756\,723$ & $ 58\,300$ & 0.82 &  15 & 1.2\% \\
         4\% & $   875$ & $   720\,743$ & $ 43\,448$ & 0.78 &  18 & 1.0\% \\\hline
        \multicolumn{7}{l}{rooster (57680)} \\\hline
        $\delta$ & time (s) & $\#T$ & $\#t$ & $q_{sr,avg}$ & $\sigma_{avg}$ & $\epsilon_{avg}$ \\\hline
        -1\% & $1\,179$ & $   807\,647$ & $ 48\,844$ & 0.85 &  15 & 2.9\% \\
         1\% & $1\,138$ & $   860\,380$ & $ 58\,616$ & 0.87 &  14 & 1.7\% \\
         2\% & $1\,014$ & $   841\,894$ & $ 32\,936$ & 0.87 &  15 & 1.5\% \\
         4\% & $   956$ & $   836\,024$ & $ 14\,096$ & 0.85 &  17 & 1.7\% \\\hline
        \multicolumn{7}{l}{lamp (61258)} \\\hline
        $\delta$ & time (s) & $\#T$ & $\#t$ & $q_{sr,avg}$ & $\sigma_{avg}$ & $\epsilon_{avg}$ \\\hline
        -1\% & $3\,640$ & $1\,923\,308$ & $317\,462$ & 0.84 &  20 & 1.8\% \\
         1\% & $5\,104$ & $1\,941\,050$ & $273\,976$ & 0.87 &  14 & 1.6\% \\
         2\% & $3\,971$ & $1\,688\,744$ & $210\,854$ & 0.83 &  15 & 1.1\% \\
         4\% & $3\,773$ & $1\,541\,226$ & $215\,350$ & 0.82 &  15 & 0.6\% \\\hline
    \end{tabular}
    \label{tab:topo-offsets-examples-table}
\end{table}

\bibliographystyle{ACM-Reference-Format}
\bibliography{literature.bib}

\appendix
\section{Quality Metrics}
\label{subsec:quality_metrics}

Throughout the offset optimization, we use 3  quality metrics to test for convergence and to determine which operations to perform.
\begin{definition}[Triangle Shape Regularity]
    The shape regularity of a triangle $t$ as defined in \cite{bank1997mesh} is its area $A(t)$ multiplied by a normalization pre-factor of $4\sqrt{3}$ and divided by the sum of squared edge lengths,
    \begin{equation*}
        \label{eq:mean_ratio_metric}
        q_{sr}(t) = 4\sqrt{3} A(t) / (l_1^2 + l_2^2 + l_3^2).
    \end{equation*}
    The shape regularity is zero for a degenerate triangle and one for an equilateral one.
\end{definition}
\begin{definition}[Triangle Normal Deviation]
    The normal deviation $\sigma(t)$ of a triangle $t$ is the maximum angle between the offset normal $n(p)$ at the triangle center $p_c$ and the offset normal at any other point $p_i$ within the triangle, excluding its boundaries,
    \begin{equation*}
        \sigma(t) = \max_{p_i \in t}(\measuredangle(n(p_c), n(p_i)).
    \end{equation*}
\end{definition}
The offset normal can be computed for any point in space by finding the projection point on the offset and normalizing the vector from the point in space to its projection. We compute the normal deviation of a triangle by comparing $n(p_c)$ with offset normals close to the vertices $p_v$ of the triangle, more precisely at the positions $p_i = 0.1 p_c + 0.9 p_v$. This choice is motivated experimentally, as using more sampling points increases the running time with negligible improvements in quality.

\begin{definition}[Offset Distance Error]
    The offset distance error of a point on the offset mesh $O$ is the absolute value of the distance of that point to the surface $S$ minus the targeted offset distance $\delta$.
\end{definition}

\end{document}